\documentclass[jmp,superscriptaddress,showpacs]{revtex4}

\usepackage{palatino}
\usepackage{mathpazo}
\usepackage{stmaryrd}

\usepackage{gastex}
\usepackage{amsfonts}
\usepackage{amssymb}
\usepackage{amsmath}
\usepackage{latexsym}
\usepackage{amsthm}
\usepackage{eepic}
\usepackage[usenames]{color}
\usepackage{hyperref}
\usepackage{graphicx}
\usepackage{mathtools}
\usepackage{enumitem}

\hypersetup{pdfpagemode=UseNone}

\newtheorem*{rep@theorem}{\rep@title}
\newcommand{\newreptheorem}[2]{%
\newenvironment{rep#1}[1]{%
 \def\rep@title{#2 \ref{##1}}%
 \begin{rep@theorem}}%
 {\end{rep@theorem}}}

\newtheorem{theorem}{Theorem}
\newtheorem{lemma}[theorem]{Lemma}
\newreptheorem{lemma}{Lemma}

\newtheorem{observation}[theorem]{Observation}
\newtheorem{cor}[theorem]{Corollary}

\theoremstyle{definition}
\newtheorem{definition}[theorem]{Definition}

\newtheorem{remark}[theorem]{Remark}

\newtheorem{example}[theorem]{Example}

\newcommand{\ket}[1]{\left| #1 \right\rangle}
\newcommand{\bra}[1]{\left\langle #1 \right|}

\newcommand{\tr}{\operatorname{Tr}}

\def\({\left(}
\def\){\right)}

\begin{document}

\title{\bf\LARGE
Minimum Entangling Power is Close to Its Maximum}

\author{
  Jianxin Chen
  \hspace{7mm}\\[3mm]
  {\small\it Alibaba Quantum Laboratory, Alibaba Group, Bellevue, Washington 98004, USA}\\[-1mm]
  \hspace{3mm}\\[1mm]
  Zhengfeng Ji
  \hspace{7mm}\\[3mm]
  {\small\it Centre for Quantum Software and Information, School of Software, Faculty of Engineering and Information Technology, University of Technology Sydney, NSW, Australia
}\\[-1mm]
  {\small\it State Key Laboratory of Computer Science, Institute of Software, Chinese Academy of Sciences, Beijing, China}\\[-1mm]
  \hspace{3mm}\\[1mm]
  David W. Kribs
  \hspace{7mm}\\[3mm]
  {\small\it Department of Mathematics \& Statistics, University of Guelph, Guelph, Ontario, Canada}\\[-1mm]
  {\small\it Institute for Quantum Computing, University of Waterloo, Waterloo, Ontario, Canada}\\[-1mm]
  \hspace{3mm}\\[1mm]
  Bei Zeng
  \hspace{7mm}\\[3mm]
  {\small\it Department of Mathematics \& Statistics, University of Guelph, Guelph, Ontario, Canada}\\[-1mm]
  {\small\it Institute for Quantum Computing, University of Waterloo, Waterloo, Ontario, Canada}\\[-1mm]
  \hspace{3mm}\\[1mm]
  Fang Zhang
  \hspace{7mm}\\[3mm]
  {\small\it Alibaba Quantum Laboratory, Alibaba Group, Bellevue, Washington 98004, USA}\\[-1mm]
}

\begin{abstract}

Given a quantum gate $U$ acting on a bipartite quantum system, its maximum (average, minimum) entangling power is the maximum (average, minimum) entanglement generation with respect to certain entanglement measure when the inputs are restricted to be product states. In this paper, we mainly focus on the ``weakest" one, i.e., the minimum entangling power, among all these entangling powers. We show that, by choosing von Neumann entropy of reduced density operator or Schmidt rank as entanglement measure, even the ``weakest" entangling power is generically very close to the maximum possible value of the entanglement measure. In other words, maximum, average and minimum entangling powers are generically close. We then study minimum entangling power with respect to other Lipschitiz-continuous entanglement measures and generalize our results to multipartite quantum systems.

As a straightforward application, a random quantum gate will almost surely be an intrinsically fault-tolerant entangling device that will always transform every low-entangled state to near-maximally entangled state.
%
%
%
%
%
\end{abstract}

\maketitle
  
\section{Introduction} \label{sec:introduction}

One of the most amazing phenomena in quantum mechanics is entanglement, which can be used to enhance channel capacity~\cite{Bennett:1992zzb}, to defeat quantum noise~\cite{Chen:2011du}, to speed up quantum algorithms~\cite{JOZSA:2003bj}. It lies in the center of many quantum information processing tasks. In addition to its important role for foundation of quantum mechanics, entanglement has also been recognized as a fundamental resource for quantum communication and quantum computation.  Considerable efforts have been devoted to various aspects of entanglement theory~\cite{Horodecki:2009gb}.

In practice, entanglement are generated through quantum evolutions. The investigation of entangling capabilities of quantum evolutions also attracts a lot of attention. Understanding various aspects of the nonlocality of quantum dynamical operations is of not only broad interest, but also fundamental importance. In \cite{Nielsen:2003jz}, Nielsen et al. proposed to develop a theory quantifying the strength of a quantum dynamical operation, as a physical resource. 
In fact, the journey to find certain strength measures capturing some nonlocal attributes of quantum dynamical operations begins even before their proposition~\cite{Zanardi:2000wg}.

A natural nonlocality measure of a quantum operation is its entangling power, which was first introduced by Zanardi et al.~\cite{Zanardi:2000wg}, given by the linear entropy produced with $U$, averaged over a certain distribution of pure product states. By replacing the linear entropy with other entanglement measures, a large family of ``average" entangling powers can be defined as follows.

For a given entanglement measure $f$, let 
\begin{equation}
P_{avg}^{(f)}\left(U_{[A:B]}\right)=\int\limits_{\ket{\alpha}\in \mathbb{C}^{d_A}, \ket{\beta}\in \mathbb{C}^{d_B}}f(U(\ket{\alpha}\otimes \ket{\beta}))
\end{equation}
be the amount of entanglement (measureed by $f$) produced with $U$, averaged over a certain distribution of pure product states. We may omit the subscript ${[A:B]}$ of $U$ if there is no ambiguity.
 
Similarly, one can also define other entangling powers as ``maximum" or ``minimum" entanglement generated by a quantum gate $U$, with respect to a given entanglement measure $f$ when the inputs are restricted to be product states:
\begin{align}
P_{max}^{(f)}\left(U_{[A:B]}\right)&=\max\limits_{\ket{\alpha}\in \mathbb{C}^{d_A}, \ket{\beta}\in \mathbb{C}^{d_B}}f(U(\ket{\alpha}\otimes \ket{\beta})),\\
P_{min}^{(f)}\left(U_{[A:B]}\right)&=\min\limits_{\ket{\alpha}\in \mathbb{C}^{d_A}, \ket{\beta}\in \mathbb{C}^{d_B}}f(U(\ket{\alpha}\otimes \ket{\beta})).
\end{align}

The above-mentioned quantities $P_{max}^{(f)}(U)$, $P_{min}^{(f)}(U)$ and $P_{avg}^{(f)}(U)$ are called entangling powers of $U$ with respect to $f$, or entangling powers for short. Such entangling powers not only provide strength measures capturing the nonlocality of a quantum operation but also play important roles in estimating communication capacity of bipartite unitary operations~\cite{Berry:2007ke}.  Investigation of entangling power attracts great interests in literature. Considerable efforts have been devoted to the ``average" entangling power~\cite{Zanardi:2000wg,Zanardi:2001vp,Yang:2008cb,Wang:2003fl,Scott:2004fc,Batle:2005vc,Balakrishnan:2010cv,Clarisse:2005by,Lu:2008wt,Ma:2007fe}. However, it is rather unclear about the other entangling powers. A further investigation of different entangling powers could help us to better understand quantum dynamics and to achieve the elusive goal of quantum computation.

Obviously, we have
\begin{equation}
0\leq P_{min}^{(f)}(U)\leq P_{avg}^{(f)}(U)\leq P_{max}^{(f)}(U)\leq \max\limits_{\ket{\psi}\in \mathbb{C}^{d_A}\otimes \mathbb{C}^{d_B}} f(\ket{\psi}).
\end{equation}

Intuitively, the maximum entangling power $P_{max}^{(f)}$ may be generically very large. In fact,  this is a direct consequence of Hayden et al.'s ``concentration of measure" phenomenon for quantum states~\cite{Hayden:2006gx} if we choose $f$ as the von Neumann entropy of reduced density operator. It was shown that a random state in $d_A\otimes d_B$ quantum system will be almost surely highly entangled if $\min{(d_A,d_B)}$ is large enough, meaning that with large probability, the state will have near-maximum entropy of entanglement.  Therefore, ``highly entangled" is a generic property of quantum states. They further proposed a subspace of dimension $\Omega\left(\frac{d_Ad_B}{{(\log d_A)}^{2.5}}\right)$ in which all states are closed to maximally entangled states in the usual topology.  Let's fix $\ket{\alpha_0}$ and $\ket{\beta_0}$.

\begin{equation}
P_{max}^{(f)}(U)=\max\limits_{\ket{\alpha}\in \mathbb{C}^{d_A}, \ket{\beta}\in \mathbb{C}^{d_B}}f(U(\ket{\alpha}\otimes \ket{\beta}))\geq f(U(\ket{\alpha_0}\otimes \ket{\beta_0})).
\end{equation}
When $U$ is chosen uniformly according to the Haar measure in the unitary group, $U(\ket{\alpha_0}\otimes \ket{\beta_0})$ is a random state according to a unique unitarily invariant probability measure which is induced by the Haar measure. It is already shown that $f\left(\ket{\phi}^{AB}\right)$ is generically very large for random quantum state $\ket{\phi}^{AB}$ when the bipartite system is large enough. Therefore, the maximum entangling power $P_{max}^{(f)}$ is generically very large for large enough systems.

If we choose $f$ as the von Neumann entropy of reduced density operator and allow ancillary systems~\cite{Linden:2009cq},  Linden et al. proved that average entangling power $P_{avg}^{(f)}$ is also generically very large when the system is large enough. They further demonstrated that ``concentration of measure" phenomenon occurs for such variant of average entangling power.

For the minimum entangling power $P_{min}^{(f)}$, although it has already been proposed as a nonlocality measure for a long time~\cite{Nielsen:2003jz}, the situation remains rather unclear.  One may even doubt whether minimum entangling power is sufficient to fully capture the nonlocality of quantum dynamics as there might exist some product state such that $f(U(\ket{\alpha}\otimes\ket{\beta})$ vanishes. Indeed, it is proved that $P_{min}^{(f)}(U)$ vanishes for every quantum gate $U$ acting on $2\otimes N$ or $3\otimes 3$ bipartite systems~\cite{Chen:2008eq}. However, the existence of quantum gates with non-vanishing minimum entangling power $P_{min}^{(f)}$ was first reported in \cite{Chen:2008eq}. It is still unclear whether such non-vanishing minimum entangling ability holds ``generically".

Informally, for a generic quantum gate $U$ acting on $\mathbb{C}^{d_A}\otimes \mathbb{C}^{d_B}$, its average entangling power is ``asymptotically" equal to maximum entangling power (denoted by $\approx$), i.e., $\lim\limits_{\min{(d_A,d_B)}\rightarrow +\infty}\frac{P_{avg}^{(f)}(U)}{P_{max}^{(f)}(U)}=1$. We will have the following relations.
\begin{equation}
0\overset{?}{\approx} P_{min}^{(f)}(U)\overset{?}{\approx} P_{avg}^{(f)}(U)\approx P_{max}^{(f)}(U)\approx \max\limits_{\ket{\psi}\in \mathbb{C}^{d_A}\otimes \mathbb{C}^{d_B}} f(\ket{\psi}).
\end{equation}

The main purpose of this paper is to provide a further investigation on the minimum entangling power. More specifically, we will demonstrate the ``concentration of measure" phenomenon for minimum entangling power with respect to a wide range of entanglement measures. Consequently, a generic quantum gate will have non-zero minimum entangling power with respect to arbitrary entanglement measure $f$.

Moreover, we will illustrate a very surprising phenomenon: in many cases, $P_{min}^{(f)}$ is very close to the maximum possible value of the entanglement measure $f$. More precisely, we prove the following three results:
\begin{enumerate}
\item [1.] The minimum entangling power $P_{min}^{(f)}(U)$ of a generic quantum gate in a $d_A\otimes d_B$ quantum system with respect to von Neumann entropy of reduced density operator will be very close to its maximum possible value with high probability when $\min{(d_A,d_B)}$ is large enough. 
\item [2.] The minimum entangling power $P_{min}^{(f)}(U)$ of a generic quantum gate in a $d_A\otimes d_B$ quantum system with respect to other Lipschitz-continuous entanglement measure $f$ will also be very close to the median value of $f$ with high probability when $\min{(d_A,d_B)}$ is large enough.
\item [3.] If we consider minimum entangling power with respect to a discrete entanglement measure, Schmidt rank,  we will have an even stronger result. $P_{min}^{(f)}(U)$ of a generic quantum gate in a $d_A\otimes d_B$ quantum system will be a constant large number with unit probability except for some degenerate cases.
\end{enumerate}

Following from the simple fact $P_{max}^{(f)}(U)\geq P_{avg}^{(f)}(U) \geq P_{min}^{(f)}(U)$, if a quantum gate $U$ has large ``minimum" entangling power $P_{min}^{(f)}(U)$, its ``average" entangling power ${E}_{A:B}^{(f)}(U)$ and ``maximum" entangling power $P_{max}^{(f)}(U)$ will be automatically large too.

We further investigate the case of multipartite quantum system, the minimum entangling power of $d_1\otimes d_2\otimes \cdots \otimes d_N(N\geq 3)$ multipartite quantum system by showing that a random quantum gate acting on this system will almost surely have near-maximum minimum entangling power in any bipartite cut.  If we take as multipartite entanglement measure tensor rank of multipartite quantum system which is monotonically decrease under SLOCC (stochastic local operations and classical communication). In general, determine or even estimate the tensor rank of an arbitrary multipartite state is very hard. This indicates the difficulty of estimating the entangling power of a specified quantum gate. Fortunately,  a simplified theory is then proposed to the generic entangling powers with respect to tensor rank. We will illustrate that the minimum entangling power with respect to tensor rank will almost surely be greater than $\left\lceil \frac{\prod\limits_{i=1}^Nd_i-\sum\limits_{i=1}^Nd_i+N}{\sum\limits_{i=1}^Nd_i-N+1}\right\rceil$. Again, this is very close to its maximum possible value.

To summarize, our results indicate a random quantum gate will be highly nonlocal in the sense that even the minimum entangling power will be very large. 

As a straightforward application,  a random quantum gate will be a universal entangling device which will always produce highly entangled state no matter what product state the input is. Our results also offer a random-pick strategy towards explicit examples.

Furthermore, with some modifications, a random quantum gate will not only transform every product state to entangled state, but also transform every state with low entanglement to a highly entangled state. In some sense, such universal entangling devices we proposed are intrinsically fault-tolerant. 

The rest of this paper is organized as follows. Section~\ref{sec:notation} introduces necessary concepts and notation from Riemann geometry and algebraic geometry which will be required for our further investigation. In Section~\ref{sec:bipartite}, we study the minimum entangling power of quantum gates acting on bipartite quantum systems. This section is further divided into two subsections. In Subsection~\ref{bipartite:entropy}, by taking as entanglement measure von Neumann entropy of reduced density operator, we first prove minimum entangling power is concentrated over the set of quantum gates for large quantum systems by combining Hayden et al.'s result of ``concentration of entropy" and the standard net argument. Then we provide a general treatment to ``concentration of measure" phenomenon over a Riemann manifold. As a consequence, we provide a stronger version of concentration of minimum entangling power over unitary group. Furthermore, by replacing von Neumann entropy of reduced density operators with any Lipschitz-continuous entanglement measure,  the corresponding minimum entangling power is also concentrated over unitary group. After dealing with continuous entanglement measures, in Subsection~\ref{bipartite:Schmidt}, we look into a discrete measure, Schmidt rank, which is the invertible SLOCC invariant measure. From an algebraic geometric point of view, the set of quantum states with bounded Schmidt rank is a determinantal variety. Hence, the study of entangling power with respect to Schmidt rank is equivalent to the study of determinantal variety. Concentration of minimum entangling power with respect to Schmidt rank is then derived. Then we look into the multipartite setting in Section~\ref{sec:multipartite}. A natural generalization of minimum entangling power to multipartite quantum system is investigated. We show that a random quantum gate acting on this system will almost surely have near-maximum entangling power in any bipartite cut. We deal with von Neumann entropy of reduced density operator and Schmidt rank in Subsection~\ref{multipartite:entropy} and ~\ref{multipartite:Schmidt} respectively. In Subsection~\ref{multipartite:tensor}, we choose as multipartite entanglement measure tensor rank which is monotonically decrease over SLOCC. Again, by observing that the set of multipartite quantum states with bounded tensor rank is a subset of a secant variety, we show a random multipartite quantum gate will has large entangling power with unit probability. Finally, in Section~\ref{sec:conclusion}, we summarize our results and provide some open problem.

\section{Notations and Preliminaries} \label{sec:notation}

This section defines notations and provides necessary background materials that will be used for later.

The set of all complex numbers is denoted as $\mathbb{C}$.  From now on, we will work throughout over field $\mathbb{C}$ of complex numbers unless otherwise specified. However, the reader should note that some results here apply over any field, not just $\mathbb{C}$. 

Every quantum system has an associated $d$-dimensional complex Hilbert space $\mathbb{C}^d$. Sometimes, we may use the associated Hilbert space to denote the quantum system if there is no ambiguity. A pure quantum state in this system is represented by a non-zero vector in $\mathbb{C}^d$, denoted as $\ket{\phi}\in \mathbb{C}^d$.

For any positive integer $n$,  the set of all $n$-tuples from $\mathbb{C}$ is called $n$-dimensional $\textit{affine space}$ over $\mathbb{C}$. An element of $\mathbb{C}^n$ is called a point, and if point $P=(a_1,a_2,\cdots,a_n)$ with $a_i\in \mathbb{C}$, then the $a_i$'s are called the coordinates of $P$.  Informally, an affine space is what is left of a vector space after forgetting its origin.

An arbitrary quantum state of a quantum system which is associated with $\mathbb{C}^d$ can be written as $\ket{\phi}=\sum\limits_k a_k\ket{k}$. $\{\ket{k}\}$ is a base of $\mathbb{C}^d$ here.

We define \textit{projective $n$-space}, denoted by $\mathbb{P}^n$, to be the set of equivalence classes of $(n+1)-$tuples $(a_0,\cdots,a_n)$ from $\mathbb{C}$, not all zero, under the equivalence relation given by $(a_0,\cdots,a_n)\sim(\lambda a_0,\cdots,\lambda a_n)$ for all $\lambda \in \mathbb{C}$, $\lambda\neq 0$.

Note that for any complex number $c\neq 0$, $\ket{\phi}\in \mathbb{C}^d$ and $c \ket{\phi}$ will represent the same state of quantum system associated with $\mathbb{C}^d$, hence the quantum state $\ket{\phi}$ actually corresponds to a point in $\mathbb{P}^{d-1}$. We will use the homogeneous coordinates $[a_1: a_2: \cdots: a_d]$ to denote the point in $\mathbb{P}^{d-1}$ that corresponds to $\ket{\phi}=\sum\limits_k a_k\ket{k}$.

Given two individual quantum systems $A$ and $B$ associated with Hilbert spaces $\mathbb{C}^{d_A}$ and $\mathbb{C}^{d_B}$ respectively, the new Hilbert space which captures the interaction of the two parties is $\mathbb{C}^{d_A}\otimes \mathbb{C}^{d_B}$, the tensor product of individual Hilbert spaces. 

A quantum state $\ket{\phi}^{AB}$ in $\mathbb{C}^{d_A}\otimes \mathbb{C}^{d_B}$ is a product state if $\ket{\phi}^{AB}=\ket{\phi}^A\otimes \ket{\phi}^B$ for some $\ket{\phi}^A\in \mathbb{C}^{d_A}$ and $\ket{\phi}^B\in \mathbb{C}^{d_B}$. Otherwise, it is called an entangled state.

A crucial observation that will be used repeatedly in this paper is that, the set of product states in composite system associated with $\mathbb{C}^{d_A}\otimes \mathbb{C}^{d_B}$ is isomorphic to a projective variety in $\mathbb{P}^{d_Ad_B-1}$, a well studied object in algebraic geometry.

We attempt to explain the observation with minimal concepts. For more details, please refer to \cite{Hartshorne:1983we}.

The \textit{polynomial ring} in $n$ variables, denoted by $\mathbb{C}[x_1,x_2,\cdots,x_n]$, is the set of polynomials in $n$ variables with coefficients in field $\mathbb{C}$.

A subset $Y$ of $\mathbb{C}^n$ is an \textit{algebraic set} if it is the common zeros of a finite set of polynomials $f_1,f_2,\cdots,f_r$ with $f_i\in \mathbb{C}[x_1,x_2,\cdots,x_n]$ for $1\leq i\leq r$, which is also denoted by $Z(f_1,f_2,\cdots,f_r)$.

One may observe that the union of a finite number of algebraic sets is an algebraic set, and the intersection of any family of algebraic sets is again an algebraic set. Therefore, by taking the open subsets to be the complements of algebraic sets, we can define a topology, called the \textit{Zariski topology} on $\mathbb{C}^n$.

A nonempty subset $Y$ of a topological space $X$ is called \textit{irreducible} if it cannot be expressed as the union of two proper closed subsets. The empty set is not considered to be irreducible.

An \textit{affine algebraic variety} is an irreducible closed subset of $C^n$, with respect to the induced topology.

A notion of algebraic variety may also be introduced in projective spaces, called projective algebraic variety: a subset $Y$ of $\mathbb{P}^n$ is an \textit{algebraic set} if it is the common zeros of a finite set of homogeneous polynomials $f_1,f_2,\cdots,f_r$ with $f_i\in \mathbb{C}[x_0,x_1,\cdots,x_n]$ for $1\leq i\leq r$. We call open subsets of irreducible projective varieties as quasi-projective varieties.

Now let's look into the following embedding.

\begin{definition}[Segre embedding and Segre variety]\label{definition:segre}
The \textit{Segre embedding} is defined as the map:
\begin{displaymath}\sigma: \mathbb{P}^{m-1} \times \mathbb{P}^{n-1} \rightarrow
\mathbb{P}^{mn-1}\end{displaymath} taking a pair of points $([x],[y])\in
\mathbb{P}^{m-1}\times \mathbb{P}^{n-1}$ to their product
\begin{equation*}
\begin{split}
\sigma: ([x_0:x_1:\cdots:x_{m-1}],[y_0:y_1:\cdots:y_{n-1}]) \\
\longmapsto [x_0y_0:x_0y_1:\cdots:x_{m-1}y_{n-1}].
\end{split}
\end{equation*}
The image of the map
is a variety, called \textit{Segre variety}, written as $\Sigma_{m-1,n-1}$.
\end{definition}

What concerns us is that Segre variety $\Sigma_{d_A-1,d_B-1}$ represents the set of product states in a bipartite quantum system $\mathbb{C}^{d_A}\otimes \mathbb{C}^{d_B}$~\cite{Brody:2001ha, Miyake:2003ft, Heydari:2005kt, Chen:2008eq}. This simple observation provides an algebraic geometric description to product states and entangled states. 

As we have already mentioned, quantum entanglement has come to be recognized as a fundamental resource that may be used for perform quantum information tasks. A natural but important question arises immediately. How much entanglement is contained in a given quantum state? Various quantities have been proposed in the last twenty years, such as the entanglement of distillation, the entanglement cost, the relative entropy of entanglement, entanglement of formation, the squashed entanglement, Schmidt rank~\cite{Plenio:2007ve}. However, many entanglement measures will coincide if we only look into the pure state case.

One particular measure is the von Neumann entropy of reduced density operators, $f\left(\ket{\phi}^{AB}\right)=S_v\left(\ket{\phi}^{AB}\right)\equiv S(\tr_B(\ket{\phi}\bra{\phi}))$, where $S(\rho)=-\tr(\rho\log \rho)$ is the von Neumann entropy which extends classical entropy to the field of quantum mechanics.  For bipartite pure states, it is the unique measure of entanglement~\cite{Popescu:1997vk}, or specifically, it is the only function on state space that satisfies certain axioms.

Another entanglement measure of broad interest is the Schmidt rank, $f\left(\ket{\phi}^{AB}\right)=SR\left(\ket{\phi}^{AB}\right)\equiv \min\left\{r: \ket{\phi}^{AB} = \sum\limits_{i=1}^r \lambda_i \ket{\alpha_i}^A\otimes \ket{\beta_i}^B\right\}$, the invertible SLOCC invariant measure for pure states.

The observation that Segre variety represents the set of product states can be straightforwardly generalized as the following.

\begin{observation}
For any integer $r$, the set of pure states with Schmidt rank no more than $r$ in composite system associated with $\mathbb{C}^{d_A}\otimes \mathbb{C}^{d_B}$ is isomorphic to a determintal variety in $\mathbb{P}^{d_Ad_B-1}$, another well studied object in algebraic geometry.
\end{observation}

The so-called determinantal variety $\Sigma_{d_A,d_B}^r$ is defined as the space of $d_A \times d_B$ matrices with some given upper bound on their ranks. It is the natural generalization of Segre variety.

To see it is also characterized as common zeros of homogenous polynomials.  The set of states with Schmidt rank $r$ in a given bipartite system $\mathbb{C}^{d_A}\otimes \mathbb{C}^{d_B}$ is isomorphic to the set of $d_A\times d_B$ matrices with rank $r$. A matrix $M$ has rank less than $r$ if and only if all its $r\times r$ minors are zero, thus $\Sigma_{d_A,d_B}^r$ is just the set of common zeros of all $r\times r$ minors. 

Hence we will say a determinantal variety is just the set of states with bounded Schmidt rank  for convenience. Through the rest of this paper, we will use $\Sigma_{d_A,d_B}^r$ to denote the set of states with Schmidt rank no more than $r$ in $\mathbb{C}^{d_A}\otimes \mathbb{C}^{d_B}$ .

\begin{definition}[Projective Determinantal Variety]

\begin{displaymath}
\Sigma_{d_A,d_B}^r=\left\{\ket{\psi}:\ket{\psi}\in \mathbb{C}^{d_A}\otimes \mathbb{C}^{d_B}, \mathop{\mathop{SR}}(\ket{\psi})\leq r\right\}
\end{displaymath}
is an irreducible projective variety.
This variety is characterized by the vanishing of all $(r+1)\times (r+1)$-minors of the state vector coefficients when written as a matrix.
\end{definition}
The projective dimension of $\Sigma_{d_A,d_B}^r$ is $d_Ad_B-(d_A-r)(d_B-r)-1$.

For $r=1$, one recovers the Segre variety.

In the multipartite setting, a pure state $\ket{\psi}_{1,2,\cdots,N}$ in $N$-partite system (associated with Hilbert space $\mathcal{H}_{1,2,\cdots,N}=\mathcal{H}_1\otimes \mathcal{H}_2\otimes \cdots\otimes \mathcal{H}_N$) is a product state (or a fully $N$-particle separable state) if and only if it can be written as
\begin{equation}
\ket{\psi}_{1,2,\cdots,N}=\ket{\psi}_1\otimes \ket{\psi}_2\otimes \cdots \otimes \ket{\psi}_N,\qquad \ket{\psi}_i\in \mathcal{H}_i,\quad i=1,2,\cdots, N.
\end{equation}

We say an $N$-partite state is bi-separable (or bi-product) if it is a product state in some bipartite cut. An $N$-partite state is a genuine entangled state if and only there does not exist a cut, against which the state is a product state, or equivalently, it is not bi-separable. For any given index set $\Gamma$ satisfying $\emptyset \subsetneq \Gamma\subsetneq \{1,2,\cdots,N\}$, let $\Sigma_{\Gamma}=\{\ket{\psi}\otimes \ket{\phi}:\ket{\psi}\in \otimes_{i\in \Gamma}\mathcal{H}_i, \ket{\phi}\in \otimes_{j \in {\Gamma}^c}\mathcal{H}_j\}$. ${\Gamma}^c$ is the complement of $\Gamma$ in $\{1,2,\cdots,N\}$. The set of genuine entangled states can be characterized by the complement of $\bigcup\limits_{\emptyset \subsetneq \Gamma\subsetneq \{1,2,\cdots,N\}} \Sigma_{\Gamma}$.

Here we will introduce another entanglement measure for multipartite quantum states, the \emph{tensor rank}, which refers to the number of product states needed to express a given multipartite quantum state.

A multipartite quantum state is said to have \emph{border rank} $r$ if it can be written as the limit of tensor rank $r$ quantum states.

Note the set of multipartite quantum states of rank at most $r$ is not closed, and by definition the set of tensors of border rank at most $r$ is the Zariski closure of this set. 

These concepts are also well studied in algebraic geometry. The $r$-th secant variety of Segre variety $\Sigma_{d_1,\cdots,d_N}$ is the Zariski closure of the union of the linear spanned by collections of $r+1$ points on Segre variety, denoted as $Sec_r(\Sigma_{d_1,\cdots,d_N})$. $Sec_r(\Sigma_{d_1,\cdots,d_N})$ is irreducible and consists of all multipartite states with border rank at $\leq r+1$.

Let $GL(n)$ and $\mathcal{U}(n)$ be the $n\times n$ complex general linear group and unitary group respectively. It is well known that the unitary group $U(n)$ is a Lie group of dimension $n^2$, i.e. a smooth manifold as well as a group, so it has a unique bi-invariant probability measure, Haar measure~\cite{Lee:2003tw}.


\begin{definition}
A set $N$ in a smooth finite dimensional manifold $M$ is said to be of measure zero if for every admissible chart $U$, $\phi$, the set $\phi(N\cap U)$ has Lebesgue measure zero in $\mathbb{R}^n$ where $\dim M=n$.
\end{definition}

\begin{definition}
A topological space $X$ is called Noetherian if it satisfies the descending chain condition for closed subsets: for any sequence $Y_1\supseteq Y_2\supseteq\cdots$ of closed subsets, there is an integer $r$ such that $Y_r=Y_{r+1}=\cdots$.
\end{definition}

\begin{theorem}[Projective Dimension Theorem, \cite{Hartshorne:1983we}]\label{lemma:dim}
Let $Y$, $Z$ be varieties of dimensions r, s in $\mathbb{P}^n$. Then every irreducible component of $Y\cap Z$ has dimension $\geq r+s-n$. Furthermore, if $r+s-n\geq 0$, then $Y\cap Z$ is nonempty.
\end{theorem}


\section{Bipartite Entangling Power}\label{sec:bipartite}

In this section, we study the minimum entangling power of quantum gates acting on bipartite quantum systems. This section is further divided into two subsections. In Subsection~\ref{bipartite:entropy}, we studied the minimum entangling power with respect to Lipschitz-continuous entanglement measures. We first prove minimum entangling power with respect to von Neumann entropy of reduced density operator is concentrated over the set of quantum gates for large quantum systems by combining Hayden et al.'s result of ``concentration of entropy"  and the standard net argument. Then we provide a general treatment to ``concentration of measure" phenomenon over a Riemann manifold. As a result, we provide a stronger version of concentration of minimum entangling power over unitary group. Furthermore, by replacing von Neumann entropy of reduced density operators with any Lipschitz-continuous entanglement measure,  the corresponding minimum entangling power is also concentrated over unitary group. 

After dealing with continuous entanglement measures, in Subsection~\ref{bipartite:Schmidt}, we look into a discrete measure, Schmidt rank, which is the invertible SLOCC invariant measure. From an algebraic geometric point of view, the set of quantum states with bounded Schmidt rank is a determinantal variety. Hence, the study of minimum entangling power with respect to Schmidt rank is equivalent to the study of determinantal variety. Concentration of minimum entangling power with respect to Schmidt rank is then derived.

\subsection{Lipschitz-continuous entanglement measures}\label{bipartite:entropy}

We first look into the minimum entangling power with respect to von Neumann entropy of reduced density operator here. Later, we will generalize our results to other Lipschitz-continuous entanglement measures.

\subsubsection{Beyond Concentration of Entanglement}\label{bipartite:entropy:Hayden}
By choosing the von Neumann entropy of reduced density operator as entanglement measure, we define the corresponding minimum entangling power as the following:
\begin{equation}
P_{min}^{(S_v)}(U)=\min\limits_{\ket{\alpha}\in \mathbb{C}^{d_A}, \ket{\beta}\in \mathbb{C}^{d_B}} S\left(\tr_B\left(U\ket{\alpha\beta}\bra{\alpha\beta}U^{\dagger}\right)\right).
\end{equation}

\begin{theorem}[Concentration of Entropy, Theorem III.3 \cite{Hayden:2006gx}]\label{thm:concentration}
Let $\ket{\phi}$ be a random state in $\mathbb{C}^{d_A}\otimes \mathbb{C}^{d_B}$, with $d_B\geq d_A\geq 3$. Then
\begin{equation}
Pr\left\{S(\tr_B{\ket{\phi}\bra{\phi}})<\log d_A-\frac{d_A}{d_B \ln 2}-\alpha\right\}\leq \exp{\left(-\frac{(d_Ad_B-1) \alpha^2}{8\pi^2 \ln 2(\log d_A)^2}\right)}.
\end{equation}
\end{theorem}

Inspired by the concentration phenomenon of entropy stated above, one may suggest that, a random gate $U$ almost surely has large minimum entangling power if the image of $U$ acting on the set of product states generically does not contain any exceptional state. If such image can always be embedded to some subspace with appropriately small dimension, then we may finalize our argument by recalling a random subspace contains only highly-entangled states.

However, the set of product states in $\mathbb{C}^{d_A}\otimes \mathbb{C}^{d_B}$ can not be embedded to a subspace with dimension smaller than $d_Ad_B$, though it can be parameterized by using only $d_A+d_B$ variables. As we described in Section~\ref{sec:notation}, the set of product states can be characterized by a set of homogeneous quadratic polynomials. We will go into it more deeply in Section~\ref{bipartite:Schmidt}. 

Here we will show that a random quantum gate $U$ will asymptotically almost surely have large minimum entangling power, by using a standard concentration and net argument. 

To start our investigation, the so-called $\epsilon$-net is required.  The concept of $\epsilon$-net was originally introduced in \cite{HAUSSLER:1987wr}, and it is then widely used in computational geometry and approximation algorithms~\cite{KOMLOS:1992wd,CHAZELLE:1995uq,BRONNIMANN:1995vv,Alon:2010kr}.
\begin{definition}
Let $X$ be a set with probability measure $\mu$, let $F$ be a collection of $\mu$-measurable subsets of $X$. For any real number $\epsilon\in (0, 1]$, a subset $Y\subseteq X$ is called an $\epsilon$-net for $(X,F)$ if for all $S\in F$, $\mu(S)\geq \epsilon$ implies $Y\bigcap S \neq \emptyset$. 
\end{definition}

Hayden et al. brings the $\epsilon$-net to the quantum information community in \cite{Hayden:2004um}.

\begin{lemma}[Lemma II.4 \cite{Hayden:2004um}]
For $0<\epsilon <1$ and $\dim \mathcal{H}=d$ there exists a set $\mathcal{N}$ of pure states in $\mathcal{H}$ with $|\mathcal{N}|\leq \left(\frac{5}{\epsilon}\right)^{2d}$, such that for every pure state $\ket{\phi}\in \mathcal{H}$ there exists $\ket{\tilde{\phi}}\in \mathcal{N}$ with $\Vert \ket{\phi}\bra{\phi}-\ket{\tilde{\phi}}\bra{\tilde{\phi}}\Vert_1\leq \epsilon$ and $\Vert \ket{\phi}-\ket{\tilde{\phi}}\Vert_2\leq \frac{\epsilon}{2}$. Such a set $\mathcal{N}$ is called an $\epsilon$-net.
\end{lemma}

For bipartite product states, there also exists an $\epsilon$-net as the following lemma.

\begin{lemma}[Lemma III.7 of \cite{Hayden:2006gx}]\label{lemma:net_Schmidt}
For $0<\epsilon <1$, the set of product states in $\mathbb{C}^{d_A}\otimes \mathbb{C}^{d_B}$ has an $\epsilon$-net $\mathcal{N}$ of size $|\mathcal{N}|\leq \left(\frac{10}{\epsilon}\right)^{2(d_A+d_B)}$.
\end{lemma}

The basic idea of involving $\epsilon$-net is natural. To show a given function $f$ is bounded over some set $M$, we only need to check $f$ is bounded by certain number over an $\epsilon$-net for $M$ if $f$ satisfies certain bounded slope condition.

For minimum entangling power $P_{min}^{(S_v)}$, we will show that it is Lipschitiz-continuous. 
\begin{lemma}
$P_{min}^{(S_v)}(U)$ is Lipschitz-continuous over unitary group $\mathcal{U}(d_Ad_B)$.
\end{lemma}
\begin{proof}
For any gates $U_1, U_2 \in \mathcal{U}(d_Ad_B)$, without loss of generality, let's say $P_{min}^{(S_v)}(U_1)$ and $P_{min}^{(S_v)}(U_2)$ achieve their minimum at points $\ket{\alpha}\otimes \ket{\beta}$ and $\ket{\gamma}\otimes \ket{\delta}$ respectively. We may further assume $P_{min}^{(S_v)}(U_1)\geq P_{min}^{(S_v)}(U_2)$.

\begin{equation}
\begin{split}
&|P_{min}^{(S_v)}(U_1)-P_{min}^{(S_v)}(U_2)|\\
={}& S(\tr_B(U_1(\ket{\alpha}\otimes\ket{\beta})))- S(\tr_B(U_2(\ket{\gamma}\otimes\ket{\delta})))\\
\leq {}&S(\tr_B(U_1(\ket{\gamma}\otimes\ket{\delta})))- S(\tr_B(U_2(\ket{\gamma}\otimes\ket{\delta})))\\
\leq {}& \sqrt{8}\log d_A \Vert   U_1(\ket{\gamma}\otimes\ket{\delta}) -U_2(\ket{\gamma}\otimes\ket{\delta}) \Vert_2\\
\leq {}& \sqrt{8} \log d_A \sqrt{\tr(U_1-U_2)(U_1^{\dagger}-U_2^{\dagger})} \\
={}& \sqrt{8} \log d_A \Vert U_1-U_2\Vert_2.
\end{split}
\end{equation}
\end{proof}

By combining Theorem~\ref{thm:concentration} and the $\epsilon$-net argument, we have the following theorem.

\begin{theorem}\label{thm:concentrate_1}
Assume $d_B\geq d_A\geq 3$, let $U$ be a random unitary gate in $\mathcal{U}(d_Ad_B)$ according to the Haar measure, then 
\begin{equation}
\mu\left(P_{min}^{(S_v)}(U)<\log d_A-\frac{d_A}{d_B \ln 2}-\alpha\right)<\left(\frac{20\sqrt{2}\log d_A}{\alpha}\right)^{2(d_A+d_B)} \exp{\left(-\frac{(d_Ad_B-1)\alpha^2}{32\pi^2\ln{2}{(\log d_A)}^2}\right)}.
\end{equation}
\end{theorem}

\begin{proof}
Following from Lemma~\ref{lemma:net_Schmidt}, we can choose $\mathcal{N}$ as an $\epsilon$-net for bipartite product states and $|\mathcal{N}|\leq \left(\frac{10}{\epsilon}\right)^{2(d_A+d_B)}$. Here, $\epsilon$ is a small real number that will be fixed later. We have
\begin{equation}
\begin{split}
&\mu\left(\min\limits_{\ket{\alpha}\in \mathbb{C}^{d_A},\ket{\beta}\in \mathbb{C}^{d_B}}S(\tr_B(U(\ket{\alpha}\otimes\ket{\beta})))<\log d_A-\frac{d_A}{d_B \ln 2}-\alpha\right)\\
\leq {}&\mu\left(\min\limits_{\ket{\psi}\in \mathcal{N}}S(\tr_B(U\ket{\psi}))<\log d_A-\frac{d_A}{d_B \ln 2}-\alpha+\sqrt{2}\epsilon\log d_A\right)\\
\leq {}& \sum\limits_{\ket{\psi}\in \mathcal{N}}\mu\left(S(\tr_B(U\ket{\psi}))<\log d_A-\frac{d_A}{d_B \ln 2}-\alpha+\sqrt{2}\epsilon\log d_A\right)\\
\leq {}& \left(\frac{10}{\epsilon}\right)^{2d_A+2d_B} \exp\left(-\frac{(d_Ad_B-1)\left(\alpha-\sqrt{2}\epsilon\log d_A\right)^2}{8\pi^2\ln 2(\log d_A)^2}\right)\label{eq:epsilon}.
\end{split}
\end{equation}

The last inequality follows from Theorem~\ref{thm:concentration} and the fact that a random unitary distributed according to the Haar measure acting on a fixed state will produce a random state distributed according to the unitarily invariant probability measure on the pure state space. The proof can be completed by choosing $\epsilon$ as $\frac{\alpha}{2\sqrt{2}\log d_A}$. It provides an upper bound on the probability that the randomly chosen gate has minimum entangling power smaller than $\left(\log d_A-\frac{d_A}{d_B \ln 2}-\alpha\right)$. Straightforwardly, a gate with minimum entangling power at least $\left(\log d_A-\frac{d_A}{d_B \ln 2}-\alpha\right)$ exists if the upper bound presented above is strictly less than $1$.
We can further secure this by requiring
\begin{equation}
\frac{1}{\alpha^2}\ln{\frac{20\sqrt{2}\log d_A}{\alpha}}\leq \frac{d_Ad_B-1}{64\pi^2\ln{2}(d_A+d_B)(\log d_A)^2}.
\end{equation}
\end{proof}

\begin{remark}\label{remark:3933}
\begin{enumerate}
\item [1.] When $d_A$ tends to $+\infty$, a random gate will approximately almost surely have large $P_{min}^{(S_v)}$. However, a non-trivial $\alpha$ satisfying both $\log d_A-\frac{d_A}{d_B \ln 2}-\alpha>0$ and $\left(\frac{20\sqrt{2}\log d_A}{\alpha}\right)^{2(d_A+d_B)} \exp{\left(-\frac{(d_Ad_B-1)\alpha^2}{32\pi^2\ln{2}{(\log d_A)}^2}\right)}<1$ exists only if $d_A\geq 3933$.\\
\item [2.] One may improve the above result by choosing a suitable $\epsilon$. The right-hand side of Equation~\ref{eq:epsilon} achieves its minimum when $\epsilon$ satisfies the transcendental equation $\epsilon^2 \ln\frac{10}{\epsilon}=8\pi^2 \ln{2} \left(\frac{d_A+d_B}{d_Ad_B-1}\right)$. However, this improvement does not change too much.
\end{enumerate}
\end{remark}

\subsubsection{Riemann Manifold Approach}\label{bipartite:entropy:Riemann}

We proved the concentration of $P_{min}^{(S_v)}$ in the previous subsection. A random quantum gate acting on a bipartite quantum system associated with $\mathbb{C}^{d_A}\otimes \mathbb{C}^{d_B}$ will almost surely have minimum entangling power very close to $\log d_A-\frac{d_A}{d_B \ln 2}$ when $\min{(d_A,d_B)}$ tends to infinity. However, as we showed in Remark~\ref{remark:3933}, Theorem~\ref{thm:concentrate_1} can not provide any insight on minimum entangling power of quantum gates acting on $\mathbb{C}^{d_A}\otimes \mathbb{C}^{d_B}$ if $\min{(d_A,d_B)}<3933$.  In this subsection, we will improve Theorem~\ref{thm:concentrate_1} by introducing some techniques from Riemann geometry.

In general, a unitary matrix $U$ has a complex determinant $\det U$ with modulus 1 but arbitrary phase. However, when using a unitary matrix to describe a quantum gate, any constant phase factor does not change the physical effect of the gate, hence $U$ is equivalent to $U'=(\det U)^{-1/n}U$, where $\det U'=1$. Therefore, without loss of generality, we only need to look into the entangling power over special unitary group. In other words, for the functions we are interested in, concentration over unitary group is equivalent to that over special unitary group.

The special unitary group is a Lie group, i.e., a group which is also a smooth manifold. In \cite{HaarNote}, the standard logarithmic Sobolev based concentration inequalities are generalized to compact Riemann manifolds, giving concentration inequalities for probability measures on certain Lie groups, including $\mathcal{SU}(n)$.

In this subsection, we will first briefly discuss those elementary Riemann geometric concepts that are necessary for understanding the concentration of measure theorem in \cite{HaarNote}. After that, by applying that theorem to the minimum entangling power function $P_{min}^{(S_v)}$ of special unitary group, we will improve our previous concentration result in Subsection~\ref{bipartite:entropy:Hayden}.

Regarding the $N\times N$ special unitary group $\mathcal{SU}(N)$ as a manifold $M$ embedded in the ambient Euclidian space $\mathbb{C}^{N^2}$, we can think of the tangent space $T_p M$ at a point $p$ on manifold $M$ as a hyperplane that best approximates $M$ around $x$. (There are more general definitions of $T_p M$ that do not depend on an ambient Euclidian space, but we will not go into that here.) A \emph{Riemann metric} $g$ is a family of inner products, $g_p: T_p M \times T_p M \to \mathbb{R}$, defined for each tangent space $T_p M$. For our purposes, it suffices to use the inner product inherited from the the ambient Euclidian space $\mathbb{C}^{N^2}$.

Usually, we use the Riemann curvature tensor to describe the curvature of Riemann manifolds, which is formally given in terms of Levi-Civita connection and Lie bracket. One may also introduce the Ricci curvature tensor $Ric$ to measure the growth rate of the volume of metric balls in the manifold.  Here, we will not go into explicit expressions for the Riemann curvature tensor or Ricci curvature tensor in terms of the Levi-Civita connection and Lie bracket. For more information about these materials, please refer to an appropriate textbook in differential geometry~\cite{Lee:2003tw}.

The Ricci curvature tensor $Ric_p$ on any point $p \in M$ can be computed. Specifically, we have the following fact:

\begin{theorem}[\cite{HaarNote}, Proposition 3.11]\label{thm:ricci}
Let $M = \mathcal{SU}(N)$, then for each $p \in M$ and each $v \in T_p(M)$,
\begin{equation}
Ric_p(v, v) = \frac{N}{2}g_p(v, v).
\end{equation}
\end{theorem}

A useful tool for proving concentration of measure theorems is the log-Sobolev inequality and the ``Herbst argument".

\begin{definition}[Log-Sobolev inequality~\cite{HaarNote}]
We say that $(X, d, \mathbb{P})$ satisfies a \emph{log-Sobolev inequality} with constant $C>0$ if, for every locally Lipschitz $g: X \to \mathbb{R}$,
\begin{equation}
\mathrm{Ent}\left(g^2\right) \le 2C\mathbb{E}\left(|\nabla g|^2\right),
\end{equation}
where the \emph{entropy} of a function $f$, $\mathrm{Ent}(f)$, is defined as
\begin{equation}
\mathrm{Ent}(f) = \mathbb{E}[f \log(f)] - (\mathbb{E}f) \log(\mathbb{E}f),
\end{equation}
and $|\nabla g|$ is defined as
\begin{equation}
|\nabla g| = \limsup_{y\to x} \frac{|g(y)-g(x)|}{d(y, x)}.
\end{equation}
\end{definition}

\begin{theorem}[\cite{HaarNote}, Theorem 3.5]\label{thm:herbst}
Suppose that $(X, d, \mathbb{P})$ satisfies a log-Sobolev inequality with constant $C>0$. Then for every 1-Lipschitz function $F: X \to \mathbb{R}, \mathbb{E}|F| < \infty$, and for every $r \ge 0$,
\begin{equation}
\mu(|F-E_\mu F| \ge r) \le 2e^{-r^2/2C}.
\end{equation}
\end{theorem}

The reason we are interested in the Ricci curvature tensor of $\mathcal{SU}(N)$ is the so-called Bakry-\'Emery creterion.

\begin{theorem}[Bakry-\'Emery~\cite{HaarNote}]\label{thm:bakry}
Let $(M,g)$ be a compact, connected, $m$-dimensional Riemannian manifold with normalized volume measure $\mu$. Suppose that there is a constant $c>0$ such that for each $p\in M$ and each $v\in T_p M$,
\begin{equation}
Ric_p(v,v)\geq \frac{1}{c} g_p(v,v).
\end{equation}
Then $\mu$ satisfies a log-Sobolev inequality with constant $c$.
\end{theorem}

Combining Theorem~\ref{thm:ricci}, Theorem~\ref{thm:bakry}, and Theorem~\ref{thm:herbst}, we immediately get the following result.

\begin{theorem}\label{thm:su}
For any differentiable function $f: SU(N)\rightarrow \mathbb{R}$ such that for any $U_1$, $U_2\in SU(N)$, $|f(U_1)-f(U_2)|\leq |f|_L\Vert U_1-U_2\Vert_2$, we have for all $\delta \geq 0$,
\begin{equation}
\mu\left(\left|f -\int\limits_{SU(N)}f(U)d\mu(U)\right|\geq \delta\right)\leq 2 e^{-\frac{N\delta^2}{4 |f|_L^2}}.
\end{equation}
\end{theorem}

\begin{remark}
For any $\theta\in [0,2\pi)$, let's $U_0(\theta)=\mathop{diag}(e^{i\theta},1,\cdots,1)$. The unitary group $\mathcal{U}(d_Ad_B)=\bigcup \limits_\theta U_0(\theta)\mathcal{SU}(d_Ad_B)$ is an infinite union. Therefore, the generalization from $\mathcal{SU}(d_Ad_B)$ to $\mathcal{U}(d_Ad_B)$ is not as straightforward as that from $\mathcal{SO}(d_Ad_B)$ to $\mathcal{O}(d_Ad_B)$.  Fortunately, a well-defined entangling power $P_{avg}(U)$ should be invariant under arbitrary phase rotation, i.e., $P_{avg}(e^{i\theta}U)=P_{avg}(U)$. As a consequence, $P_{avg}(U)$ is concentrated over $\mathcal{SU}(d_Ad_B)$ if and only if it is concentrated over $\mathcal{U}(d_Ad_B)$. To state our result more precisely, we will keep focusing mainly on $\mathcal{SU}(d_Ad_B)$. Though, the reader should keep in mind that our discussion also holds for $\mathcal{U}(d_Ad_B)$ since functions we are mostly concerned are invariant under arbitrary phase rotation.
\end{remark}

\begin{cor}
For any $\delta>0$, $\mu\left(\left|P_{min}^{(S_v)}(U) -\int\limits_{\mathcal{SU}(d_Ad_B)}P_{min}^{(S_v)}(U)d\mu(U)\right|\geq \delta\right)\leq 2 e^{-\frac{d_Ad_B\delta^2}{32 (\log d_A)^2}}$.
\end{cor}

So the above corollary improves the concentration inequality we provided in Subsection~\ref{bipartite:entropy:Hayden}. In the rest of this Subsection, we will estimate the central point $\int\limits_{\mathcal{SU}(d_Ad_B)}P_{min}^{(S_v)}(U)d\mu(U)$. Though numerical calculation is always possible, an analytic estimation is not that easy. Fortunately, the main idea from Subsection~\ref{bipartite:entropy:Hayden} will provide us an analytic lower bound for $\int\limits_{\mathcal{SU}(d_Ad_B)}P_{min}^{(S_v)}(U)d\mu(U)$. 

Following from Theorem~\ref{thm:concentrate_1}, for any $\lambda>0$ , we have
\begin{equation}
\mu\left(P_{min}^{(S_v)}(U)\geq \log d_A-\lambda-\frac{d_A}{d_B \ln 2}\right)\geq 1-  \left(\frac{20\sqrt{2}\log d_A}{\lambda}\right)^{2d_A+2d_B} \exp\left(-\frac{(d_Ad_B-1)\lambda^2}{32\pi^2\ln 2(\log d_A)^2}\right)
\end{equation}
where $U$ is uniformly chosen at random from $\mathcal{U}(d_Ad_B)$ according to the Haar measure and $d_B\geq d_A\geq 3$.

Straightforwardly, for any $\lambda>0$, we have
\begin{equation}
\int\limits_{\mathcal{SU}(d_Ad_B)}P_{min}^{(S_v)}(U)d\mu(U) \geq \left(\log d_A-\lambda-\frac{d_A}{d_B \ln 2}\right)\left(1-  \left(\frac{20\sqrt{2}\log d_A}{\lambda}\right)^{2d_A+2d_B} \exp\left(-\frac{(d_Ad_B-1)\lambda^2}{32\pi^2\ln 2(\log d_A)^2}\right)\right).
\end{equation}

Some careful calculation would lead to the following stronger result. 
\begin{equation}\label{eq:int}
\int\limits_{\mathcal{SU}(d_Ad_B)}P_{min}^{(S_v)}(U)d\mu(U)\geq \log d_A-\frac{d_A}{d_B \ln 2}-1.
\end{equation}

The details can be found in Appendix~\ref{appendix:medium} within which the above claim is formulated as Corollary~\ref{cor:lower}.

\begin{cor}
Let's assume a bipartite quantum system $\mathbb{C}^{d_A}\otimes \mathbb{C}^{d_B}$ is given and $\min(d_A,d_B)$ is large enough. Then, with high probability, a random gate acting on this system will transform every product state to a nearly-maximally entangled state, or more precisely, to a state with entanglement entropy higher than $\left(\log d_A-\frac{d_A}{d_B \ln 2}-1\right)$. In other words, a random gate acting on this system will have $P_{min}^{(S_v)}\geq \log d_A-\frac{d_A}{d_B \ln 2}-1$.
\end{cor}

It is natural to choose the von Neumann entropy of reduced density operator as the entanglement measure, as we have done so far in this section, but it is also important to consider other entanglement measures as they might be useful in different scenarios.  It must be emphasized that our result is not dependent on the choice of entanglement measure. In the following theorem, we will show that,  minimum entangling power with respect to any Lipschitz-continuous entanglement measure is also Lipschitz-continuous. Hence, such minimum entangling power is also concentrated over the special unitary group (or equivalently, the unitary group).

%

\begin{theorem}
For any Lipschitz-continuous entanglement measure $f$, let $P_{min}^{(f)}$ be the minimum entangling power of quantum gates in $\mathcal{SU}(d_Ad_B)$, with respect to entanglement measure $f$. Then $P_{min}^{(f)}$ is concentrated on the special unitary group $SU(d_Ad_B)$. More specifically,  for all $\delta > 0$,
\begin{equation}
\mu\left(\left|P_{min}^{(f)}(U) -\int\limits_{\mathcal{SU}(d_Ad_B)}P_{min}^{(f)}(U)d\mu(U)\right|\geq \delta\right)\leq 2 e^{-\frac{d_Ad_B\delta^2}{4 |f|_L^2}}.
\end{equation}
\end{theorem}
\begin{proof}
To prove our claim, we will show that $P_{min}^{(f)}(U)$ is Lipschitz-continuous over special unitary group $\mathcal{SU}(d_Ad_B)$.
For any gates $U_1, U_2 \in \mathcal{SU}(d_Ad_B)$, without loss of generality, let's say $P_{min}^{(f)}(U_1)$ and $P_{min}^{(f)}(U_2)$ achieve their minimum at points $\ket{\alpha}\otimes \ket{\beta}$ and $\ket{\gamma}\otimes \ket{\delta}$ respectively. We can even assume $P_{min}^{(f)}(U_1)\geq P_{min}^{(f)}(U_2)$.

\begin{equation}
\begin{split}
&|P_{min}^{(f)}(U_1)-P_{min}^{(f)}(U_2)|\\
={}& f(U_1(\ket{\alpha}\otimes\ket{\beta}))- f(U_2(\ket{\gamma}\otimes\ket{\delta}))\\
\leq {}&f(U_1(\ket{\gamma}\otimes\ket{\delta}))- f(U_2(\ket{\gamma}\otimes\ket{\delta}))\\
\leq {}& |f|_L \Vert   U_1(\ket{\gamma}\otimes\ket{\delta}) -U_2(\ket{\gamma}\otimes\ket{\delta}) \Vert_2\\
\leq {}& |f|_L \sqrt{\tr(U_1-U_2)(U_1^{\dagger}-U_2^{\dagger})} \\
={}& |f|_L \Vert U_1-U_2\Vert_2.
\end{split}
\end{equation}

Therefore, $P_{min}^{(f)}(U)$ is also a Lipschitz-continuous function. 
\end{proof}

The estimation of $\int\limits_{\mathcal{SU}(d_Ad_B)}P_{min}^{(f)}(U)d\mu(U)$ can be performed similarly as done in Equation~\ref{eq:int}. 
\begin{lemma}[Levy's Lemma, \cite{Hayden:2006gx}]
Let $f: \mathbb{S}^k\rightarrow \mathbb{R}$ be a function with Lipschitz constant $|f|_L$ with respect to the Euclidean norm and a point $X\in \mathbb{S}^k$ be chosen uniformly at random. Then
\begin{enumerate}
\item [1.] $Pr\{f(X)-\mathbb{E}(f)\gtrless \pm\alpha\}\leq 2 \exp\left(-\frac{C_1(k+1)\alpha^2}{|f|_L^2}\right)$ and
\item [2.] $Pr\{f(X)-m(f)\gtrless \pm\alpha\}\leq \exp\left(-\frac{C_2(k-1)\alpha^2}{|f|_L^2}\right)$.
\end{enumerate}
for absolute constants $C_i>0$ that may be chosen as $C_1=\frac{1}{9\pi^3\ln 2}$ and $C_2=\frac{1}{2\pi^2\ln 2}$. $\mathbb{E}(f)$ is the mean value of $f$ and  $m(f)$ is a median for $f$.
\end{lemma}

By applying Levy's Lemma directly, we will have $Pr\{f(X)< \mathbb{E}(f)-\alpha\}\leq 2 \exp\left(-\frac{2C_1d_Ad_B\alpha^2}{|f|_L^2}\right)$ for Lipschitz-continuous entanglement measure $f$. By combining similar ideas from both Theorem~\ref{thm:concentrate_1} and Equation~\ref{eq:int}, we will have $\int\limits_{\mathcal{SU}(d_Ad_B)}P_{min}^{(f)}(U)d\mu(U) \geq \mathbb{E}(f)-\epsilon$ where $\epsilon$ is some tiny positive number.

For pure states, the distillable entanglement, entanglement cost, entanglement of formation, relative entropy of entanglement and squashed entanglement are all equal to von Neumann entropy of reduced density operator $S(\tr_B(\ket{\psi}_{AB}))$. Let's consider some other Lipschitz continuous entanglement measures~\cite{Plenio:2007ve}.

\begin{remark}
It is known that the positivity of the partial transpose with respect to party $B$ of a bipartite state $\rho_{AB}$ is a necessary condition for separability and  is suffice to prove the non-distillability. The negativity $N$ defined as $N(\ket{\psi}_{AB})=\frac{\left\Vert (\ket{\psi}\bra{\psi}_{AB})^{T_B}\right\Vert -1}{2}$ is an entanglement monotone that captures the negativity in the spectrum of the partial transpose. Here $\Vert \cdot \Vert$ is the trace norm. Note the trace norm is invariant under partial transpose.
\begin{equation}
\begin{split}
&N(\ket{\alpha}_{AB})-N(\ket{\beta}_{AB})\\
={}&\frac{1}{2}\left(\left\Vert (\ket{\alpha}\bra{\alpha}_{AB})^{T_B}\right\Vert-\left\Vert (\ket{\beta}\bra{\beta}_{AB})^{T_B}\right\Vert\right)\\
\leq {}&\frac{1}{2}\left\Vert (\ket{\alpha}\bra{\alpha}_{AB}-\ket{\beta}\bra{\beta}_{AB})^{T_B}\right\Vert\\
= {}&\frac{1}{2}\Vert (\ket{\alpha}\bra{\alpha}_{AB}-\ket{\beta}\bra{\beta}_{AB})\Vert\\
\leq {}& \Vert \ket{\alpha}-\ket{\beta}\Vert_2.
\end{split}
\end{equation}
Hence the negativity $N$ is $1$-Lipschitz continuous, which follows the minimum entangling power with respect to negativity is concentrated over unitary group.
\end{remark}

\begin{remark}
Let's introduce ancilla  to our concept of entangling power. For any quantum gate $U$ acting on a bipartite system $\mathbb{C}^{d_A}\otimes \mathbb{C}^{d_B}$, let's define its complete minimum entangling power as the minimal entanglement generation of $U$, acting on $\mathbb{C}^{d_{A'}}\otimes \mathbb{C}^{d_A}\otimes \mathbb{C}^{d_B}\otimes \mathbb{C}^{d_{B'}}$ in $(AA':BB')$-cut for arbitrary ancillas  $\mathbb{C}^{d_{A'}}$ and $\mathbb{C}^{d_{B'}}$, i.e.,
\begin{equation}
\begin{split}
&\tilde{P}_{min}^{(S_v)}U_{[A:B]}\\
={}&\min\limits_{A',B'}P_{min}^{(S_v)}(I_{A'}\otimes U\otimes I_{B'})_{[AA':BB']}\\
={}&\min\limits_{A',B'}\left(\min\limits_{\ket{\alpha}\in \mathbb{C}^{d_{A'}}\otimes \mathbb{C}^{d_A}, \ket{\beta}\in \mathbb{C}^{d_B}\otimes \mathbb{C}^{d_{B'}}} S\left(\tr_{BB'}\left((I_{A'}\otimes U_{AB}\otimes I_{B'})\ket{\alpha\beta}\bra{\alpha\beta}(I_{A'}\otimes U_{AB}^{\dagger}\otimes I_{B'})\right)\right)\right).
\end{split}
\end{equation}

We will show $\tilde{P}_{min}^{(S_v)}(U)$ is also Lipschitz-continuous over unitary group $\mathcal{U}(d_Ad_B)$. Without loss of generality, we may assume $d_{A'}=d_A$ and $d_{B'}=d_B$.

For any gates $U_1, U_2 \in \mathcal{U}(d_Ad_B)$,  let's say $\tilde{P}_{min}^{(S_v)}(U_1)$ and $\tilde{P}_{min}^{(S_v)}(U_2)$ achieve their minimum at points $\ket{\alpha}_{A'A}\otimes \ket{\beta}_{BB'}$ and $\ket{\gamma}_{A'A}\otimes \ket{\delta}_{BB'}$ respectively. We can further assume $\tilde{P}_{min}^{(S_v)}(U_1)\geq \tilde{P}_{min}^{(S_v)}(U_2)$.

\begin{equation}
\begin{split}
&\left|\tilde{P}_{min}^{(S_v)}(U_1)-\tilde{P}_{min}^{(S_v)}(U_2)\right|\\
={}& S(\tr_{BB'}((I_{A'}\otimes U_1\otimes I_{B'})(\ket{\alpha}\otimes\ket{\beta})))- S(\tr_{BB'}((I_{A'}\otimes U_2\otimes I_{B'})(\ket{\gamma}\otimes\ket{\delta})))\\
\leq {}& S(\tr_{BB'}((I_{A'}\otimes U_1\otimes I_{B'})(\ket{\gamma}\otimes\ket{\delta})))- S(\tr_{BB'}((I_{A'}\otimes U_2\otimes I_{B'})(\ket{\gamma}\otimes\ket{\delta})))\\
\leq {}& \sqrt{8}\log (d_Ad_{A'}) \Vert   ((I_{A'}\otimes U_1\otimes I_{B'})(\ket{\gamma}\otimes\ket{\delta})) -((I_{A'}\otimes U_2\otimes I_{B'})(\ket{\gamma}\otimes\ket{\delta})) \Vert_2\\
\leq {}& 4\sqrt{2} \log d_A \sqrt{\tr(I_{A'}\otimes U_1\otimes I_{B'}-I_{A'}\otimes U_2\otimes I_{B'})(I_{A'}\otimes U_1^{\dagger}\otimes I_{B'}-I_{A'}\otimes U_2^{\dagger}\otimes I_{B'})} \\
={}& 4\sqrt{2d_Ad_B} \log d_A \Vert U_1-U_2\Vert_2.
\end{split}
\end{equation}
\end{remark}

%
%

\subsection{Schmidt Rank As Entanglement Measure}\label{bipartite:Schmidt}
Another important entanglement measure is the Schmidt rank which is defined as the number of Schmidt coefficients of a bipartite state. However, our approach proposed in the previous subsection can not deal with discrete entanglement measures like Schmidt rank. In this subsection, we will provide an even stronger concentration result for minimum entangling power with respect to Schmidt rank, by using some techniques from algebraic geometry.

By taking as entanglement measure the Schmidt rank here, we define the corresponding minimum entangling power as the following:
\begin{equation}
P_{min}^{(SR)}(U)=\min\limits_{\ket{\alpha}\in \mathbb{C}^{d_A}, \ket{\beta}\in \mathbb{C}^{d_B}} \mathop{SR}(U\ket{\alpha\beta}).
\end{equation}

\begin{theorem}\label{thm:Schmidt}
A random unitary gate $U\in \mathcal{U}(d_Ad_B)$ will almost surely have minimum entangling power $P_{min}^{(SR)}(U)$ exactly equal to $\left\lceil \frac{d_A+d_B-\sqrt{(d_A-d_B)^2+4(d_A+d_B)-8}}{2}\right\rceil$. Moreover, no quantum gate has $P_{min}^{(SR)}(U)$ large than this number.
\end{theorem}
\begin{proof}
The proof will be divided into three steps. 

\begin{enumerate}
\item [1.] $P_{min}^{(SR)}(U)\leq  r_0=\left\lceil \frac{d_A+d_B-\sqrt{(d_A-d_B)^2+4(d_A+d_B)-8}}{2}\right\rceil$ for any $U\in \mathcal{U}(d_Ad_B)$.
\item [2.] There does exist some $U\in \mathcal{U}(d_Ad_B)$ such that $P_{min}^{(SR)}(U) \geq r_0$.
\item [3.] The set of unitary gates with minimum entangling power less than its maximum value $r_0$ has measure zero in the unitary group $\mathcal{U}(d_Ad_B)$.
\end{enumerate}

Firstly,  observe that a unitary gate $U$ has minimum entangling power $P_{min}^{(SR)}\geq r$ if and only if $U$ will map the set of product states (or equivalently, the Segre variety $\Sigma_{d_A,d_B}$) to a set of entangled states with Schmidt rank at least $r$. In other words, $U(\Sigma_{d_A,d_B})\bigcap \Sigma_{d_A,d_B}^{r-1}=\emptyset$.

Assume we can find some gate $U\in \mathcal{U}(d_Ad_B)$ with minimum entangling power $P_{min}^{(SR)}(U) >  r_0$,  
\begin{equation}
\begin{split}
&\dim U(\Sigma_{d_A,d_B})+\dim \Sigma_{d_A,d_B}^{r_0}\\
={}&d_A+d_B-2+d_Ad_B-(d_A-r_0)(d_B-r_0)-1\\
\geq{}& d_Ad_B-1 \\
={}&\dim \mathbb{P}^{d_Ad_B-1}.
\end{split}
\end{equation}
 
By applying Theorem~\ref{lemma:dim},  $U(\Sigma_{d_A,d_B})\bigcap \Sigma_{d_A,d_B}^{r_0-1}$ is not empty. 

Therefore, $P_{min}^{(SR)}(U)\leq r_0$ for any $U\in \mathcal{U}(d_Ad_B)$.

Secondly,  there is at least  some $U\in \mathcal{U}(d_Ad_B)$ such that $P_{min}^{(SR)}(U) \geq r_0$. In order to prove this claim, let's consider the set of unitary gates with minimum entangling power no more than $r_0-1$.  Let $\mathcal{U}_{r_0}=\left\{\Phi|\Phi\in \mathcal{U}(d_Ad_B),\Phi(\Sigma_{d_A,d_B})\cap \Sigma_{d_A,d_B}^{r_0-1}\neq \emptyset\right\}$. Our aim is to show $\mathcal{U}_{r_0}$ is a proper subset in $\mathcal{U}(d_Ad_B)$. If so, then the existence of quantum gates with maximum minimum entangling power will be automatically guaranteed.

Let's consider the Zariski topology on the projective space. In this setting, the unitary group $\mathcal{U}(d_Ad_B)$ is Zariski dense in the general linear group $GL(d_A,d_B)$~\cite{Schmitt:2008ww}. We further define $X_{r_0}=\left\{\Phi|\Phi\in GL(d_Ad_B),\Phi(\Sigma_{d_A,d_B})\cap \Sigma_{d_A,d_B}^{r_0-1}\neq \emptyset\right\}$.  $X_{r_0}$ contains all quantum gates with minimum entangling power less than $r_0$. Dimension of its  Zariski closure $\dim \overline{X_{r_0}}$ is bounded by $d_A^2d_B^2-(d_A-r_0)(d_B-r_0)+2r_0-3$.

The proof of the above bound is quite technical; we defer it to Appendix~\ref{appendix:closure}. 

Now we prove the existence of quantum gate $U$ with $P_{min}^{(SR)}(U)$ at least $r_0$ as follows. If it does not exist, $\mathcal{U}(d_Ad_B)\subset X_{r_0}$, then $GL(d_Ad_B)=\overline{\mathcal{U}(d_Ad_B)}\subset{\overline{X_{r_0}}}$. However,  $\dim(\overline{X_{r_0}})\leq d_A^2d_B^2-(d_A-r_0)(d_B-r_0)+2r_0-3 < d_A^2d_B^2 = \dim(GL(d_Ad_B))$. It's a contradiction. So $\mathcal{U}(d_Ad_B)\not\subset X_{r_0}$, i.e. a unitary operator $\Phi\in\mathcal{U}(d_Ad_B)$ with  $P_{min}^{(SR)}(U) \geq r_0$ exists.
According to the previous result, we have $P_{min}^{(SR)}(U) = r_0$.

Thirdly, we will now show $\mathcal{U}_{r_0}$ is not only a proper subset, but also a neglectable subset in $\mathcal{U}(d_Ad_B)$.

$\mathcal{U}(d_Ad_B)$ is a locally compact Lie group of dimension $d_A^2d_B^2$. Recall that $\dim(\overline{X_{r_0}})$ is at most $d_A^2d_B^2-(d_A-r_0)(d_B-r_0)+2r_0-3<d_A^2d_B^2=\dim(\mathcal{U}(d_Ad_B))$. 

We have shown $\dim(\overline{X_{r_0}}) < d_A^2d_B^2 = \dim(\mathcal{U}(d_Ad_B))$.  $\overline{X_{r_0}}$ is Neotherian, then $\overline{X_{r_0}}$ is union of finite many smooth subvariesties of $GL(d_Ad_B)$ with lower dimensions. Hence $\overline{X_{r_0}}\cap \mathcal{U}(d_Ad_B)$ (which contains $X_{r_0}\cap \mathcal{U}(d_Ad_B)$, the set
of our main interests) is union of finite many submanifolds of $\mathcal{U}(d_Ad_B)$ with lower dimensions. Apply Morse-Sard theorem, $X_{r_0} \cap \mathcal{U}(d_Ad_B)$ is measure zero in $\mathcal{U}(d_Ad_B)$ which implies that a random unitary operator $U$ almost surely has $P_{min}^{(SR)}(U)=r_0$.
\end{proof}

\begin{cor}
$P_{min}^{(SR)}(U)=0$ for any $U\in \mathcal{U}(d_Ad_B)$ if and only if $\min{(d_A,d_B)}\leq 2$ or $(d_A,d_B)=(3,3)$.
\end{cor}

\begin{lemma}
For any quantum gate $U$, we have $P_{min}^{(S_v)}(U)\leq \log P_{min}^{(SR)}(U)  \leq \log \left\lceil \frac{d_A+d_B-\sqrt{(d_A-d_B)^2+4(d_A+d_B)-8}}{2}\right\rceil$.
\end{lemma}
\begin{proof}
For any $\ket{\alpha}$ and $\ket{\beta}$, we have
$S(\tr_B(U\ket{\alpha\beta}\bra{\alpha\beta}U^{\dagger}))\leq \log (\mathop{rank} \tr_B(U\ket{\alpha\beta}\bra{\alpha\beta}U^{\dagger}))= \log \mathop{SR} (U\ket{\alpha\beta})$.

Therefore, $P_{min}^{(S_v)}(U)\leq \log P_{min}^{(SR)}(U)$ for any $U\in \mathcal{U}(d_Ad_B)$.

By applying Theorem~\ref{thm:Schmidt}, we have $P_{min}^{(S_v)}(U)\leq \log P_{min}^{(SR)}(U)  \leq \log \left\lceil \frac{d_A+d_B-\sqrt{(d_A-d_B)^2+4(d_A+d_B)-8}}{2}\right\rceil$.

\end{proof}

\begin{cor}
In the case $d_A=d_B=d$, we have
\begin{enumerate}
\item [1.] a random unitary gate $U\in \mathcal{U}(d_Ad_B)$ will almost surely have $P_{min}^{(SR)}(U)=\left\lceil d-\sqrt{2(d-1)}\right\rceil$.
\item [2.] for any $U\in \mathcal{U}(d_Ad_B)$, $P_{min}^{(SR)}(U)\leq \left\lceil d-\sqrt{2(d-1)}\right\rceil$ and $P_{min}^{(S_v)}(U)\leq \log d-\frac{\sqrt{2(d-1)}-1}{d}$.
\end{enumerate}
\end{cor}

We also show how similar ideas can be applied to prove an extension of Theorem~\ref{thm:Schmidt}.
\begin{theorem}\label{thm:Schmidt2}
A random quantum gate acting on $\mathbb{C}^{d_A}\otimes \mathbb{C}^{d_B}$ will almost surely map every (product or entangled) state with Schmidt rank no more than $r$ to entangled state with Schmidt rank at least $\left\lceil \frac{d_A+d_B-\sqrt{(d_A-d_B)^2+4r(d_A+d_B)-4(r^2+1)}}{2}\right\rceil$ for any integer $r\leq \min\{d_A,d_B\}$.
\end{theorem}
\begin{proof}
Observe the fact that $U$ will map the set of low rank $(\leq r)$ states (or equivalently, the Segre variety $\Sigma_{d_A,d_B}^r$) to a set of high rank $(\geq s+1)$entangled states is equivalent to $U(\Sigma_{d_A,d_B}^r)\bigcap \Sigma_{d_A,d_B}^{s}=\emptyset$.

According to Theorem~\ref{lemma:dim}, if $\dim \Sigma_{d_A,d_B}^r+\dim \Sigma_{d_A,d_B}^s\geq d_Ad_B-1$, then $U(\Sigma_{d_A,d_B}^r)\bigcap \Sigma_{d_A,d_B}^{s}\neq \emptyset$. 

Similarly, let's define $X_{r,s}=\left\{\Phi|\Phi\in GL(d_Ad_B),\Phi(\Sigma_{d_A,d_B}^r)\cap \Sigma_{d_A,d_B}^{s}\neq \emptyset\right\}$. Dimension of its  Zariski closure $\dim \overline{X_{r,s}}$ is bounded by $d_A^2d_B^2-(d_Ad_B-1)+\dim \Sigma_{d_A,d_B}^r+\dim \Sigma_{d_A,d_B}^s$. If $\dim \overline{X_{r,s}}< \dim {GL}(d_Ad_B)$, then we must have $\mathcal{U}(d_Ad_B)\not\subseteq \overline{X_{r,s}}$. Otherwise, $\mathcal{U}(d_Ad_B)\subseteq \overline{X_{r,s}}$ will lead to $GL(d_Ad_B)=\overline{\mathcal{U}(d_Ad_B)}\subseteq \overline{X_{r,s}}$ which contradicts $\dim \overline{X_{r,s}}< \dim {GL}(d_Ad_B)$.

The largest $s$ satisfying $\dim \Sigma_{d_A,d_B}^r+\dim \Sigma_{d_A,d_B}^s < d_Ad_B-1$ is $\left\lceil \frac{d_A+d_B-\sqrt{(d_A-d_B)^2+4r(d_A+d_B)-4(r^2+1)}}{2}\right\rceil-1$. 

The randomness can be proved in a very similar way as the proof of Theorem~\ref{thm:Schmidt}.

\end{proof}

So far, we provide the existence of quantum gates with nonvanishing minimum entangling power for bipartite quantum systems. However, the proof of existence is not constructive and it doesn't provide any thoughts on the structure of quantum gates with nonvanishing minimum entangling power. However, following from the genericness of entangling power, the following observation is natural.

\begin{observation}
To construct a concrete totally entangling gate, we can choose a quantum gate randomly~\cite{Mezzadri:2007uu} and then verify whether it has nonzero entangling power. Since the large entangling power is a generic property of quantum gates, a concrete example will be provided with unit probability.
\end{observation}

So, the only thing left is, how to verify whether a given quantum gate has nonvanishing minimum entangling power?

Note that, for a given gate $U$, $U$ has totally entangling power if and only if the equation $U(\ket{\alpha}\otimes \ket{\beta})=\ket{\gamma}\otimes \ket{\delta}$ do not have non-zero solution $(\ket{\alpha},\ket{\beta},\ket{\gamma},\ket{\delta})$. This can be verified algorithmically by using Gr\"{o}bner basis reduction~\cite{Buchberger:2001wo}.

\begin{example}
Here is an example of quantum gate acting on $\mathbb{C}^3\otimes \mathbb{C}^4$ with positive minimum entangling power. $+$ and $-$ means $+1$ and $-1$ respectively. It is a Hadamard matrix of order $12$. Note that $12$ is the smallest dimension such that quantum gates with $P_{min}^{(SR)}>0$ exist. 
\begin{equation}
U_H=\left(
\begin{array}{cccccccccccc}
+&-&-&-&-&-&-&-&-&-&-&-\\
+&+&-&+&-&-&-&+&+&+&-&+\\
+&+&+&-&+&-&-&-&+&+&+&-\\
+&-&+&+&-&+&-&-&-&+&+&+\\
+&+&-&+&+&-&+&-&-&-&+&+\\
+&+&+&-&+&+&-&+&-&-&-&+\\
+&+&+&+&-&+&+&-&+&-&-&-\\
+&-&+&+&+&-&+&+&-&+&-&-\\
+&-&-&+&+&+&-&+&+&-&+&-\\
+&-&-&-&+&+&+&-&+&+&-&+\\
+&+&-&-&-&+&+&+&-&+&+&- \\
+&-&+&-&-&-&+&+&+&-&+&+ 
\end{array}
\right).
\end{equation}
\end{example}

Following from Theorem~\ref{thm:Schmidt}, there is no quantum gate $U$ acting on $\mathbb{C}^3\otimes \mathbb{C}^4$ with minimum entangling power at least $3$. The example we provide here is the one with largest minimum entangling power in $\mathbb{C}^3\otimes \mathbb{C}^4$.

Unfortunately, though we can do some computation over computational algebraic system to verify whether a given quantum gate has nonvanishing minimum entangling power theoretically, the verification requires an exponentially increasing amount of resources (e.g., time, computational memory), it is practically impossible to verify whether a quantum gate $U\in \mathcal{U}(d_Ad_B)$ has nonvanishing minimum entangling power for large $d_A$ and $d_B$.

Here, let's introduce ancillary systems to our minimum entangling power.  We will show that $P_{min}^{(SR)}(U_{AB}\otimes I_{A'B'})_{[AA':BB']}\geq P_{min}^{(SR)}(U_{AB})$. As a consequence, a quantum gate $U$ acting on $\mathbb{C}^{d_A}\otimes \mathbb{C}^{d_B}$ with nonvanishing minimum entangling power will automatically lead to nonvanishing minimum entangling power of $U\otimes I_{A'}\otimes I_{B'}$ acting on $\mathbb{C}^{d_{A'}}\otimes \mathbb{C}^{d_A}\otimes \mathbb{C}^{d_B}\otimes \mathbb{C}^{d_{B'}}$ in $(AA':BB')$-cut.

For any $\ket{\alpha}_{AA'}$ and $\ket{\beta}_{BB'}$, we have expansions that $\ket{\alpha}_{AA'}=\sum\limits_i \ket{\alpha_i}_A\ket{i}_{A'}$ and $\ket{\beta}_{BB'}=\sum\limits_j\ket{\beta_j}_B\ket{j}_{B'}$. Let's choose $i'$ and $j'$ such that $\ket{\alpha_{i'}}\neq 0$ and $|\beta_{j'}\rangle\neq 0$.

\begin{equation}
\bra{i'j'}_{A'B'}(U_{AB}\otimes I_{A'B'})\ket{\alpha}_{AA'}\ket{\beta}_{BB'}=\bra{i'j'}_{A'B'}\sum\limits_{i,j}U_{AB}\left(\ket{\alpha_i}_A\ket{\beta_j}_B\right)\ket{ij}_{A'B'}=U_{AB}\left(\ket{\alpha_{i'}}_A|\beta_{j'}\rangle_B\right).
\end{equation}

Therefore, $\mathop{SR}((U_{AB}\otimes I_{A'B'})\ket{\alpha}_{AA'}\ket{\beta}_{BB'})\geq \mathop{SR}\left(U_{AB}\left(\ket{\alpha_{i'}}_A|\beta_{j'}\rangle_B\right)\right)\geq P_{min}^{(SR)}(U_{AB})$ for  any $\ket{\alpha}_{AA'}$ and $\ket{\beta}_{BB'}$. It follows that $P_{min}^{(SR)}(U_{AB}\otimes I_{A'B'})_{[AA':BB']}\geq P_{min}^{(SR)}(U_{AB})$.

The above observation implies that, to construct quantum gates with nonvanishing minimum entangling power for all non-degenerate bipartite quantum system $\mathbb{C}^{d_A}\otimes \mathbb{C}^{d_B}$, we only need to construct such quantum gates  for prime numbers $d_A$ and $d_B$. 

At the end of this subsection, we will illustrate that  certain family of Householder-type quantum gates has nonvanishing minimum entangling power by introducing some appropriately chosen subspace. 

For any given $\ket{\psi}\in \mathbb{C}^{d_A}\otimes \mathbb{C}^{d_B}$, let $U_{\psi}=I-2\ket{\psi}\bra{\psi}$ be the associated Householder-type unitary matrix. For any $\ket{\alpha}\in \mathbb{C}^{d_A}$ and $\ket{\beta}\in \mathbb{C}^{d_B}$,  $U_{\psi}(\ket{\alpha}\otimes \ket{\beta})=\ket{\alpha}\otimes \ket{\beta}-2\bra{\psi}\alpha\beta\rangle\ket{\psi}$. If we further assume that $\ket{\psi}$ is chosen to satisfy $\mathop{SR}(\ket{\psi})\geq 3$, then $U_{\psi}(\ket{\alpha}\otimes \ket{\beta})$ is always entangled for any $\ket{\alpha}\ket{\beta}\not\perp \ket{\psi}$.

Though the unitary gate $U_{\psi}$ given above can not entangle $\ket{\alpha}\otimes \ket{\beta}$ when $\ket{\alpha}\otimes \ket{\beta} \perp \ket{\psi}$, it suggests the following observation.

\begin{lemma}\label{lemma:Householder}
Let $U=I-2P$, where $P=\sum\limits_{i=1}^k \ket{\psi_k}\bra{\psi_k}$ is a projection to some subspace $\mathcal{H}_P\subset \mathbb{C}^{d_A}\otimes \mathbb{C}^{d_B}$. If $\mathcal{H}_P$ can be appropriately chosen to satisfy the following conditions:
\begin{enumerate}
\item [1] for any $\ket{\psi}\in \mathcal{H}_P$, $\mathop{SR}(\ket{\psi})\geq r+1$;
\item [2] for any $\ket{\phi}\in \mathcal{H}_P^{\perp}$, $\mathop{SR}(\ket{\phi})\geq 2$,
\end{enumerate}
Then $U$ is a quantum gate with minimum entangling power at least $r$.
\end{lemma}
\begin{proof}
We just follow the lines of the above arguments. For any $\ket{\alpha}\in \mathbb{C}^{d_A}$ and $\ket{\beta}\in \mathbb{C}^{d_B}$,  $U(\ket{\alpha}\otimes \ket{\beta})=\ket{\alpha}\otimes \ket{\beta}-2P(\ket{\alpha}\otimes \ket{\beta})$. With appropriately chosen $P$ satisfying conditions stated above,  $P(\ket{\alpha}\otimes \ket{\beta})$ is nonzero and it has Schmidt rank at least $r+1$ which follows $\mathop{SR}(U(\ket{\alpha}\otimes \ket{\beta}))\geq r$.
\end{proof}

The existence of subspace $P$ is guaranteed by the following result from \cite{Walgate:2008be}.

\begin{lemma}\label{lemma:Walgate}
Let $\mathbb{C}^{d_A}\otimes \mathbb{C}^{d_B}$ be a bipartite Hilbert space, where $d_B\geq d_A$. Then for almost all subspaces $s$-dimensional subspace $S\subseteq \mathbb{C}^{d_A}\otimes \mathbb{C}^{d_B}$, the number of states with Schmidt rank $r$ or less contained in $S$ is exactly
\begin{displaymath}
\left\{
\begin{array}{ll}
0, &\textrm{if $s\leq s_{max}^{'}=(d_A-r)(d_B-r)$;}\\
\prod\limits_{j=0}^{d_1-r-1}\frac{(d_B+r)!r!}{(r+j)!(d_B-r+j)!},& \textrm{if $s=s_{max}^{'}+1$;}\\
\infty, & \textrm{otherwise.}
\end{array}\right.
\end{displaymath}
\end{lemma}

Lemma~\ref{lemma:Walgate} says almost all subspaces with dimension no more than $(d_A-r)(d_B-r)$ is completely void of states with Schmidt rank $r$ or less. 

If there is no such $P$ satisfying conditions stated in Lemma~\ref{lemma:Householder}, then for any $(d_A-r)(d_B-r)$-dimensional subspace, if it is completely void of states with Schmidt rank $\leq r$, its complementary must contain some product states. 

However, when $d_Ad_B-(d_A-r)(d_B-r)\leq (d_A-1)(d_B-1)$, a random subspace with dimension $(d_A-r)(d_B-r)$ is almost surely completely void of states with Schmidt rank $\leq r$, and its complementary subspace is almost surely completely void of product states. This is a contradiction.

So the only requirement for the existence of Householder-type quantum gates with nonvanishing minimum entangling power is $d_Ad_B-(d_A-r)(d_B-r)\leq (d_A-1)(d_B-1)$.

\begin{remark}
Comparing with Theorem~\ref{thm:concentration_entropy}, there are two major differences. (1.) Theorem~\ref{thm:concentration_entropy} illustrated  some ``concentration of measure" phenomenon over high-dimensional unitary group. It is unclear whether such phenomenon still occurs in small-dimensional systems. The discrete version of ``concentration of measure" phenomenon we presented in Theorem~\ref{thm:Schmidt} occurs in any bipartite quantum system (except for some degenerate cases). (2.)  In Theorem~\ref{thm:concentration_entropy}, a random quantum gate has large minimum entangling power with high probability. In the discrete version, a random quantum gate has large minimum entangling power with unit probability.
\end{remark}


\section{Multipartite Entangling Power}\label{sec:multipartite}

In Section~\ref{sec:bipartite}, we studied the minimum entangling power of a quantum gate in a bipartite cut of some quantum system. Here, we will look into the multipartite quantum systems. In the case of quantum system composed of $N\geq 3$ subsystems, the structure of entangled states is much more complicated than that in the bipartite case. 

In the multipartite setting, a pure state $\ket{\psi}_{1,2,\cdots,N}$ in $N$-partite system (associated with Hilbert space $\mathcal{H}_{1,2,\cdots,N}=\mathcal{H}_1\otimes \mathcal{H}_2\otimes \cdots\otimes \mathcal{H}_N$) is a product state (or a fully $N$-particle separable state) if and only if it can be written as
\begin{equation}
\ket{\psi}_{1,2,\cdots,N}=\ket{\psi}_1\otimes \ket{\psi}_2\otimes \cdots \otimes \ket{\psi}_N.
\end{equation}

However, the violation of the above condition doesn't imply a ``truly" $N$-partite entanglement. For instance, one may consider a tripartite state $\ket{\phi}_{123}=\ket{\Phi}_{12}\otimes \ket{\phi}_3$ where $\ket{\Phi}_{12}$ is a Bell state and $\ket{\phi}_3$ is some qubit state.

In the multipartite setting, we say an $N$-partite state is bi-separable (or bi-product) if it is a product state in some bipartite cut. An $N$-partite state is a genuine entangled state if and only there does not exist a bipartite cut, against which the state is a product state, or equivalently, it is not bi-separable. For any given index set $\Gamma$ satisfying $\emptyset \subsetneq \Gamma\subsetneq \{1,2,\cdots,N\}$, let $\Sigma_{\Gamma}=\l\{\ket{\psi}\otimes \ket{\phi}:\ket{\psi}\in \otimes_{i\in \Gamma}\mathcal{H}_i, \ket{\phi}\in \otimes_{j \in {\Gamma}^c}\mathcal{H}_j\}$. The set of genuine entangled states can be characterized by the complement of $\bigcup\limits_{\emptyset \subsetneq \Gamma\subsetneq \{1,2,\cdots,N\}} \Sigma_{\Gamma}$.

The concept of minimum entangling power can be easily generalized to quantum gates acting on multipartite quantum systems. We say a quantum gate acting on multipartite quantum system $\mathcal{H}=\mathcal{H}_1\otimes \mathcal{H}_2\otimes \cdots \otimes \mathcal{H}_N(N\geq 3)$ has minimum entangling power with respect to entanglement measure $f$ as the following:
\begin{equation}
\begin{split}
&P_{min}^{(f)}(U_{1:2:\cdots:N})\\
={}&\min\limits_{\ket{\psi}_i\in \mathcal{H}_i, i=1,2,\cdots, N} \left(\min\limits_{\emptyset \subsetneq \Gamma\subsetneq \{1,2,\cdots,N\}}  f_{\Gamma:\Gamma^c}(U(\ket{\psi}_1\otimes \ket{\psi}_2\otimes \cdots \otimes \ket{\psi}_N))\right)\\
={}&\min\limits_{\emptyset \subsetneq \Gamma\subsetneq \{1,2,\cdots,N\}}\left(\min\limits_{\ket{\psi}_i\in \mathcal{H}_i, i=1,2,\cdots, N}   f_{\Gamma:\Gamma^c}(U(\ket{\psi}_1\otimes \ket{\psi}_2\otimes \cdots \otimes \ket{\psi}_N))\right)\\
\geq {}&\min\limits_{\emptyset \subsetneq \Gamma\subsetneq \{1,2,\cdots,N\}}P_{min}^{(f)}\left(U_{[\Gamma:\Gamma^c]}\right),
\end{split}
\end{equation}
i.e., minimum entangling power of a multipartite quantum gate is greater than or equal to the minimum possible value of its ``bipartite" minimum entangling power in any bipartite cut.

In Subsection~\ref{multipartite:entropy}, we will study the multipartite minimum entangling power with respect to von Neumann entropy of reduced density operators. In Subsection~\ref{multipartite:Schmidt}, by taking as entanglement measure the minimum Schmidt rank in any bipartite cut, we will generalize our result on $P_{min}^{(SR)}$ to multipartite setting. In Subsection~\ref{multipartite:tensor}, we will introduce another well-studied multipartite entanglement measure, which is called tensor rank. We will study the corresponding minimum entangling power with respect to tensor rank. 

\subsection{von Neumann Entropy of Reduced Density Operator As Multipartite Entanglement Measure}\label{multipartite:entropy}

In the multipartite setting,  a quantum gate $U$ acting on multipartite quantum system $\mathcal{H}=\mathcal{H}_1\otimes \mathcal{H}_2\otimes \cdots \otimes \mathcal{H}_N(N\geq 3)$ has minimum entangling power
\begin{equation}
P_{min}^{(S_v)}(U)\geq \min\limits_{\emptyset \subsetneq \Gamma\subsetneq \{1,2,\cdots,N\}}P_{min}^{(S_v)}\left(U_{[\Gamma:\Gamma^c]}\right).
\end{equation}

Let's recall our result on bipartite minimum entangling power with respect to entropy of entanglement,  we have the following claim:

Let's temporarily fix certain nonempty $\Gamma\subsetneq \{1,2,\cdots,N\}$. Denote the dimension of the entire $N$-partite system as $d_0 \equiv \prod\limits_{i=1}^N d_i$, and define $d_{\Gamma} \equiv \prod\limits_{i\in \Gamma}d_i$. For any $\delta>0$,

\begin{equation}
\begin{split}
{}&\mu\left(P_{min}^{(S_v)}\left(U_{[\Gamma:\Gamma^c]}\right) \leq \int\limits_{\mathcal{SU}(d_0)}P_{min}^{(S_v)}\left(U_{[\Gamma:\Gamma^c]}\right)d\mu(U)- \delta\right) \\
\leq{}&\mu\left(\left|P_{min}^{(S_v)}\left(U_{[\Gamma:\Gamma^c]}\right) -\int\limits_{\mathcal{SU}(d_0)}P_{min}^{(S_v)}\left(U_{[\Gamma:\Gamma^c]}\right)d\mu(U)\right|\geq \delta\right) \\
\leq{}& 2 e^{-\frac{d_0\delta^2}{32(\log d_\Gamma)^2}}.
\end{split}
\end{equation}

Therefore, we will have
\begin{equation}
\mu\left(\mathop{\bigwedge}\limits_{\Gamma} \left(P_{min}^{(S_v)}\left(U_{[\Gamma:\Gamma^c]}\right) > \int\limits_{\mathcal{SU}(d_0)}P_{min}^{(S_v)}\left(U_{[\Gamma:\Gamma^c]}\right)d\mu(U)- \delta\right)\right)\geq 1- 2 \sum\limits_{\Gamma}e^{-\frac{d_0\delta^2}{32(\log d_\Gamma)^2}}.
\end{equation}

In both sides of the above inequality, $\Gamma$ runs over all nonempty subsets of $\{1,2,\cdots,N\}$ with $d_\Gamma\leq d_{\Gamma^c}$.

Recall that, when further assume $d_\Gamma\leq  d_{\Gamma^c}$, then the central point $\int\limits_{\mathcal{SU}(d_0)}P_{min}^{(S_v)}\left(U_{[\Gamma:\Gamma^c]}\right)d\mu(U)\geq \sum\limits_{i\in \Gamma}\log d_i-\frac{ d_\Gamma}{ d_{\Gamma^c} \ln 2}-1$.

Therefore, we will obtain the following theorem by applying $P_{min}^{(S_v)}(U)\geq \min\limits_{\emptyset \subsetneq \Gamma\subsetneq \{1,2,\cdots,N\}}P_{min}^{(S_v)}\left(U_{[\Gamma:\Gamma^c]}\right)$.

\begin{theorem}
Assume $d_1\leq d_2\leq \cdots \leq d_N$, a quantum gate $U\in \mathcal{U}(d_1\cdots d_N)$ is chosen randomly according to the Haar measure. Then, for any $\delta>0$, 
\begin{equation}
\mu\left(P_{min}^{(S_v)}(U) >\log d_1-\frac{ d_1}{ \prod\limits_{i=2}^Nd_i \ln 2}-1-\delta \right)\geq 1- 2 \sum\limits_{\emptyset \subsetneq \Gamma\subsetneq \{1,2,\cdots,N\}\atop s.t. d_\Gamma\leq  d_{\Gamma^c}}e^{-\frac{d_0\delta^2}{32(\log d_\Gamma)^2}}.
\end{equation}
\end{theorem}

Note $e^{-\frac{d_0\delta^2}{32(\log d_\Gamma)^2}}$ is negligible when $\delta$ is small enough. A finite sum of negligible quantities is negligible. This implies a random quantum gate acting on multipartite quantum system will also have large entangling power with high probability.

\subsection{Schmidt Rank As Multipartite Entanglement Measure} \label{multipartite:Schmidt}

By taking as entanglement measure the minimum Schmidt rank in any bipartite cut here, we define the corresponding minimum entangling power as the following:
\begin{equation}
P_{min}^{(SR)}(U)=\min\limits_{\ket{\psi}_i\in \mathcal{H}_i, i=1,2,\cdots, N} \min\limits_{\emptyset \subsetneq \Gamma\subsetneq \{1,2,\cdots,N\}}  {\mathop{SR}}_{\Gamma:\Gamma^c}(U(\ket{\psi}_1\otimes \ket{\psi}_2\otimes \cdots \otimes \ket{\psi}_N)).
\end{equation}

\begin{theorem}
Assuming $d_1\leq d_2\leq \cdots \leq d_N$, a random quantum gate $U\in \mathcal{U}(d_1\cdots d_N)$ will almost surely have multipartite entangling power $P_{min}^{(SR)}(U)=\left\lceil \frac12\left(d_1+\prod\limits_{i=2}^N d_i-\sqrt{\left(d_1-\prod\limits_{i=2}^N d_i\right)^2+4\left(d_1+\prod\limits_{i=2}^N d_i\right)-8}\right)\right\rceil$.
\end{theorem}

In fact, we have a stronger version for multipartite entangling power.
\begin{theorem}
A random quantum gate $U\in \mathcal{U}(d_1\cdots d_N)$ will almost surely have bipartite entangling power $P_{min}^{(SR)}\left(U_{[\Gamma:\Gamma^c]}\right)=\left\lceil \frac12\left(d_\Gamma+d_{\Gamma^c}-\sqrt{\left(d_\Gamma-d_{\Gamma^c}\right)^2+4\left(d_\Gamma+d_{\Gamma^c}\right)-8}\right)\right\rceil$ in $(\Gamma:\Gamma^c)$-cut for any $\emptyset \subsetneq \Gamma\subsetneq \{1,2,\cdots,N\}$.
\end{theorem}

\begin{proof}
Following from Theorem~\ref{thm:Schmidt}, for a given $\emptyset \subsetneq \Gamma\subsetneq \{1,2,\cdots,N\}$, the set of quantum gates with bipartite entangling power $P_{min}^{(SR)}(U_{[\Gamma:\Gamma^c]})<\left\lceil \frac12\left(d_\Gamma+d_{\Gamma^c}-\sqrt{\left(d_\Gamma-d_{\Gamma^c}\right)^2+4\left(d_\Gamma+d_{\Gamma^c}\right)-8}\right)\right\rceil$ has measure zero in the unitary group.

Therefore, the set of quantum gates with bipartite entangling power $P_{min}^{(SR)}(U_{[\Gamma:\Gamma^c]})<\left\lceil \frac12\left(d_\Gamma+d_{\Gamma^c}-\sqrt{\left(d_\Gamma-d_{\Gamma^c}\right)^2+4\left(d_\Gamma+d_{\Gamma^c}\right)-8}\right)\right\rceil$ for some $\Gamma$ is a finite union of measure zero sets, hence it also has measure zero in $U(d_1\cdots d_N)$. This completes our proof.

\end{proof}

\subsection{Tensor Rank As Multipartite Entanglement Measure} \label{multipartite:tensor}
In the previous two subsections, we studied the minimum entangling power of quantum gates acting on multipartite quantum system which is defined as the minimum entanglement in any bipartite cut when the input is restricted to be a multipartite product state.
The minimum Schmidt rank is a natural way to quantify the entangling power of a multipartite quantum gate. Here we will introduce another kind of entanglement measure for multipartite quantum states, the tensor rank, which refers to the minimum number of product states needed to express a given multipartite quantum state. 

A multipartite quantum state is said to have tensor rank $r$ if it can be written as a linear combination of $r$ product states. 

Let's look into the state 
\begin{equation}
\ket{\Phi}=\ket{000}+\ket{001}+\ket{010}+\ket{100}.
\end{equation}

The tensor rank of $\ket{\Phi}$ is $3$. However, $\ket{\Phi}$ can be approximated as closely as one likes by a series of quantum states with tensor rank $2$, as consider:

\begin{equation}
\ket{\Phi(\epsilon)}=\frac{1}{\epsilon}((\epsilon-1)\ket{000}+(\ket{0}+\epsilon\ket{1})(\ket{0}+\epsilon\ket{1})(\ket{0}+\epsilon\ket{1})).
\end{equation}

A multipartite quantum state is said to have border rank $r$ if it can be written as the limit of tensor rank $r$ quantum states. Tensor rank and border rank are denoted by $\mathop{TR}$ and $\mathop{BR}$ respectively.

Note the set of multipartite quantum states of rank at most $r$ is not closed, and by definition the set of tensors of border rank at most $r$ is the Zariski closure of this set. 

These concepts are well studied in algebraic geometry. The $r$-th secant variety of Segre variety $\Sigma_{d_1,\cdots,d_N}$ is the Zariski closure of the union of the linear spanned by collections of $r+1$ points on Segre variety, denoted as $Sec_r(\Sigma_{d_1,\cdots,d_N})$. $Sec_r(\Sigma_{d_1,\cdots,d_N})$ is irreducible and consists of all multipartite states with border rank at most $\leq r+1$.

By taking as entanglement measure the tensor rank and border rank here, we define the corresponding entangling powers respectively as the following:
\begin{align}
P_{min}^{(TR)}(U)&=\min\limits_{\ket{\psi_1}\in \mathcal{H}_1, \cdots,\ket{\psi_N}\in \mathcal{H}_N} \mathop{TR}(U\ket{\psi_1\otimes\cdots\otimes\psi_N}),\\
P_{min}^{(BR)}(U)&=\min\limits_{\ket{\psi_1}\in \mathcal{H}_1, \cdots,\ket{\psi_N}\in \mathcal{H}_N} \mathop{BR}(U\ket{\psi_1\otimes\cdots\otimes\psi_N}).
\end{align}

It is easy to observe that
\begin{lemma}
$P_{min}^{(TR)}(U)\geq P_{min}^{(BR)}(U)$ for any $U\in \mathcal{U}(d_1\cdots d_N)$.
\end{lemma}

Hence we will first look into $P_{min}^{(BR)}$, the entangling power with respect to the border rank.

We have $P_{min}^{(BR)}(U)\geq r+2$ if and only if $U(\Sigma_{d_1,\cdots,d_N})\bigcap Sec_r(\Sigma_{d_1,\cdots,d_N})=\emptyset$.

\begin{theorem}
There is some quantum gate $U$ with entangling power
\begin{equation}
P_{min}^{(BR)}(U)\geq r+2
\end{equation}
if and only if
\begin{equation}
\dim Sec_r(\Sigma_{d_1,\cdots,d_N})< d_0-1-\sum\limits_{i=1}^N (d_i-1).
\end{equation}
\end{theorem}

\begin{proof}
We first look into the ``only if" part. Assume $P_{min}^{(BR)}(U)\geq r+2$. If $\dim Sec_r(\Sigma_{d_1,\cdots,d_N})\geq d_0-1-\sum\limits_{i=1}^N (d_i-1)$, or equivalently, $\dim U(\Sigma_{d_1,\cdots,d_N})+\dim Sec_r(\Sigma_{d_1,\cdots,d_N})\geq d_0-1$, follows from Theorem~\ref{lemma:dim}, $U(\Sigma_{d_1,\cdots,d_N})\bigcap Sec_r(\Sigma_{d_1,\cdots,d_N})\neq \emptyset$. Therefore, $P_{min}^{(BR)}(U)\leq r+1$. It's a contradiction, which proves the ``only if" part.

On the other hand, assume $\dim Sec_r(\Sigma_{d_1,\cdots,d_N})< d_0-1-\sum\limits_{i=1}^N (d_i-1)$.  Let's define
\begin{equation}
X=\left\{\Phi\in GL(d_0): \Phi(\Sigma_{d_1,\cdots,d_N})\bigcap Sec_r(\Sigma_{d_1,\cdots,d_N})\neq \emptyset\right\}.
\end{equation}
By following the lines of proof in Appendix~\ref{appendix:closure}, we will have
\begin{equation}
\dim \overline{X}\leq d_0^2-d_0+1+\dim Sec_r(\Sigma_{d_1,\cdots,d_N})+\dim \Sigma_{d_1,\cdots,d_N}.
\end{equation}

Hence, assume $P_{min}^{(BR)}(U)\geq r+2$, which follows $\mathcal{U}(d_0)\subseteq X$. Again,
\begin{equation}
\dim \overline{\mathcal{U}(d_0)}\leq \dim \overline{X}\leq  d_0^2-d_0+1+\dim Sec_r(\Sigma_{d_1,\cdots,d_N})+\dim \Sigma_{d_1,\cdots,d_N}.
\end{equation}
Recall the Zariski closure of the unitary group is exactly the general linear group. Hence
\begin{equation}
d_0^2\leq  d_0^2-d_0+1+\dim Sec_r(\Sigma_{d_1,\cdots,d_N})+\dim \Sigma_{d_1,\cdots,d_N}<d_0^2.
\end{equation}

It's a contradiction. Hence there does exist some quantum gate $U$ such that $P_{min}^{(BR)}(U)\geq r+2$.
\end{proof}

\begin{cor} \label{cor:tr}
Assume $r$ is the largest integer satisfying 
\begin{equation}
\dim Sec_r(\Sigma_{d_1,\cdots,d_N})< d_0-1-\sum\limits_{i=1}^N (d_i-1),
\end{equation}
then a random quantum gate acting on $\mathcal{H}_1\otimes \cdots\otimes \mathcal{H}_N$ will almost surely have minimum entangling power $P_{min}^{(BR)}=r+2$ and $P_{min}^{(TR)}\geq r+2$.
\end{cor}

\begin{cor}
For the case $N=2$, $Sec_r(\Sigma_{d_1,d_2})$ will coincide with the determinantal variety $\Sigma_{d_1,d_2}^{r+1}$. Hence $\dim Sec_r(\Sigma_{d_1,\cdots,d_N})=d_1d_2-(d_1-r-1)(d_2-r-1)-1$. One will recover Theorem~\ref{thm:Schmidt}.
\end{cor}

So the only thing left is to calculate the dimension of secant variety of Segre variety. This mission can be accomplished by using Terracini's Lemma. We just describe some results here without proof.  We refer the reader to \cite{Catalisano:2011fu,Aladpoosh:2011ec,Catalisano:2003kx} for details.

The expected dimension of secant variety of Segre variety $Sec_r(\Sigma_{d_1,\cdots,d_N})$ is defined as $\min\left\{d_0-1, r\left(\sum\limits_{i=1}^N (d_i-1)+1\right)+\sum\limits_{i=1}^N (d_i-1)\right\}$. The expected dimension is an upper bound for $\dim Sec_r(\Sigma_{d_1,\cdots,d_N})$. If, in some cases, $\dim Sec_r(\Sigma_{d_1,\cdots,d_N})$ is strictly smaller than the expected dimension, then the Segre variety $\Sigma_{d_1,\cdots,d_N}$ is said to be defective.

\begin{theorem}[\cite{Catalisano:2011fu,Aladpoosh:2011ec,Catalisano:2003kx}]
For $N\geq 3$, $\dim Sec_r(\Sigma_{d_1,\cdots,d_N})=\min\left\{d_0-1, r\left(\sum\limits_{i=1}^N (d_i-1)+1\right)+\sum\limits_{i=1}^N (d_i-1)\right\}$ except for the following defective cases:
\begin{enumerate}
\item [a.] $3\otimes 3\otimes 3$;
\item [b.] $2\otimes 2\otimes d \otimes d$;
\item [c.] $3\otimes d \otimes d$ for odd $d$;
\item [d.] $3\otimes 4 \otimes 4$;
\item [e.] $d_1\otimes \cdots \otimes d_m \otimes d_{m+1}$ with $\mathop{\prod}\limits_{i=1}^m d_i+1-\sum\limits_{i=1}^m(d_i-1)< d_{m+1}$. 
\end{enumerate}

For other positive cases, a random quantum gate will almost surely have minimum entangling power $P_{min}^{(BR)}=\left\lceil \frac{d_0-\sum\limits_{i=1}^Nd_i+N}{\sum\limits_{i=1}^Nd_i-N+1}\right\rceil$ and $P_{min}^{(TR)}\geq \left\lceil \frac{d_0-\sum\limits_{i=1}^Nd_i+N}{\sum\limits_{i=1}^Nd_i-N+1}\right\rceil$.
\end{theorem}

\begin{example}
Consider the $N$-qubit quantum system $\mathcal{H}^{\otimes N}$ where $\dim \mathcal{H}=2$ and $N\geq 5$, $\dim Sec_r(\mathcal{H}^{\otimes N})=\min\{2^N-1,r(N+1)+N\}$. By applying Corollary~\ref{cor:tr}, a random quantum gate acting on this system will almost surely have minimum entangling power $P_{min}^{(BR)}=\left\lceil \frac{2^N-N}{N+1}\right\rceil$ and $P_{min}^{(TR)}\geq \left\lceil \frac{2^N-N}{N+1} \right\rceil$.
\end{example}

\section{Conclusion and Open Problems}\label{sec:conclusion}

Given a quantum gate $U$ acting on a bipartite quantum system, its maximum (average, minimum) entangling power is defined as the maximum (average, minimum) entanglement generated when $U$ is restricted to be acting on product states. Such family of entangling powers will not only be practically crucial in quantifying the amount of entanglement that can be generated by quantum gates, but also theoretically important in capturing the nonlocality of quantum dynamic operations. In this paper, we mainly focus on the ``weakest" one among all these entangling powers. 

Let us summarize the various situations we know
\begin{enumerate}
  \item [1.] Given a bipartite quantum system $\mathbb{C}^{d_A}\otimes \mathbb{C}^{d_B}$, a quantum gate $U\in \mathcal{U}(d_A\cdots d_B)$ is chosen randomly according to the Haar measure. The minimum entangling power with respect to entanglement measure $f$ is defined as the following
  \begin{equation}
P_{min}^{(f)}(U)=\min\limits_{\ket{\alpha}\in \mathbb{C}^{d_A}, \ket{\beta}\in \mathbb{C}^{d_B}} f(U\ket{\alpha\beta}).
\end{equation}
  $U$ will almost surely map every product state to a near-maximally entangled state in the following three senses.
  \begin{enumerate}[label*=\arabic*.]
    \item [1.1] For any $\delta>0$, 
\begin{equation}
\mu\left(\left|P_{min}^{(S_v)}(U) -\int\limits_{\mathcal{SU}(d_Ad_B)}P_{min}^{(S_v)}(U)d\mu(U)\right|\geq \delta\right)\leq 2 e^{-\frac{d_Ad_B\delta^2}{32 (\log d_A)^2}}.
\end{equation}
Here, $\int\limits_{\mathcal{SU}(d_Ad_B)}P_{min}^{(S_v)}(U)d\mu(U)\geq \log d_A-\frac{d_A}{d_B \ln 2}-1$.
    \item [1.2] For any $\delta > 0$,
\begin{equation}
\mu\left(\left|P_{min}^{(f)}(U) -\int\limits_{\mathcal{SU}(d_Ad_B)}P_{min}^{(f)}(U)d\mu(U)\right|\geq \delta\right)< 2 e^{-\frac{d_Ad_B\delta^2}{4 |f|_L^2}}
\end{equation}
for any Lipschitz-continuous entanglement measure $f$.
    \item [1.3] $P_{min}^{(SR)}(U)=\left\lceil \frac{d_A+d_B-\sqrt{(d_A-d_B)^2+4(d_A+d_B)-8}}{2}\right\rceil$ for a generic quantum gate $U$.
  \end{enumerate}
  \item [2.] Given a multipartite quantum system $\mathcal{H}_1\otimes \mathcal{H}_2\otimes \cdots \otimes \mathcal{H}_N(N\geq 3)$ where $\mathcal{H}_1,\mathcal{H}_2,\cdots,\mathcal{H}_N$ are Hilbert spaces with dimensions $d_1\leq d_2\leq \cdots \leq d_N$ respectively, a random gate acting on this system will almost surely map every product state to a genuine entangled state if $\min\limits_{1\leq i\leq N}\dim \mathcal{H}_i\geq 3$. A quantum gate $U\in \mathcal{U}(d_1\cdots d_N)$ is chosen randomly according to the Haar measure. The minimum entangling power with respect to bipartite entanglement measure $f$ is defined as the following
\begin{equation}
P_{min}^{(f)}(U_{1:2:\cdots:N})=\min\limits_{\ket{\psi}_i\in \mathcal{H}_i, i=1,2,\cdots, N} \left (\min\limits_{\emptyset \subsetneq \Gamma\subsetneq \{1,2,\cdots,N\}}  f_{\Gamma:\Gamma^c}(U(\ket{\psi}_1\otimes \ket{\psi}_2\otimes \cdots \otimes \ket{\psi}_N))\right).
\end{equation}
If $f$ is a multipartite entanglement measure, then the minimum entangling power with respect to $f$ is defined as usual.
\begin{equation}
P_{min}^{(f)}(U_{1:2:\cdots:N})=\min\limits_{\ket{\alpha_1}\in \mathcal{H}_1, \ket{\alpha_2}\in \mathcal{H}_2, \cdots, \ket{\alpha_N}\in \mathcal{H}_N} f(U\ket{\alpha_1\alpha_2\cdots \alpha_N}).
\end{equation}

$U$ will almost surely map every product state to a near-maximally entangled state in the following three senses. 
  \begin{enumerate}[label*=\arabic*.]
  \item[2.1] For any $\delta>0$,   
\begin{equation}
\mu\left(P_{min}^{(S_v)}(U) >\log d_1-\frac{ d_1}{ \prod\limits_{i=2}^Nd_i \ln 2}-1-\delta \right)\geq 1- 2 \sum\limits_{\emptyset \subsetneq \Gamma\subsetneq \{1,2,\cdots,N\}\atop s.t. d_\Gamma\leq  d_{\Gamma^c}}e^{-\frac{d_0\delta^2}{32(\log d_\Gamma)^2}},
\end{equation}
where $d_0 = \prod\limits_{i=1}^N d_i$, and $d_{\Gamma} = \prod\limits_{i\in \Gamma}d_i$.
 \item [2.2] $P_{min}^{(SR)}(U)=\left\lceil \frac12\left(d_1+\prod\limits_{i=2}^N d_i-\sqrt{(d_1-\prod\limits_{i=2}^N d_i)^2+4(d_1+\prod\limits_{i=2}^N d_i)-8}\right)\right\rceil$.
 \item [2.3] $P_{min}^{(TR)}(U)\geq \left\lceil \frac{d_0-\sum\limits_{i=1}^Nd_i+N}{\sum\limits_{i=1}^Nd_i-N+1}\right\rceil$ except for some degenerate cases.
 \end{enumerate} 
\end{enumerate}

To summarize, we show that, for most quantum gates, even the ``weakest" entangling power is very close to its maximum possible value. In other words, the maximum, average and minimum entangling powers are close for almost every quantum gates. Our results provides a step towards a better understanding of the nonlocality of quantum dynamics.

As a straightforward application, a random quantum gate will almost surely be an intrinsically fault-tolerant entangling device which will always transform every product state to near-maximally entangled state. Furthermore, our methods can be partially modified to prove that a random quantum gate will also transform every low-entangled state to highly-entangled state. 

We showed that the minimum entangling power is generically large. One may pick up a quantum gate randomly and then verify whether it has non-vanishing minimum entangling power. However, the verification of nonvanishing minimum entangling power is equivalent to the computation of a Gr\"{o}bner basis which may require time that is exponential or even doubly exponential in the number of solutions of the polynomial system in the worst case. To construct an explicit family of such gates in a very simple form is still worth exploring.



\bigskip

\subsection*{Acknowledgements}
DWK is supported by NSERC Discovery Grant 400160, NSERC Discovery Accelerator Supplement 400233 and Ontario Early Researcher Award 048142. BZ is supported by NSERC Discovery Grant 400500 and CIFAR. 
Part of this work was done when JC was a PhD student with Prof. Mingsheng Ying in Tsinghua University. We thank Mingsheng Ying, Jun Yu, Vladim\'{i}r Bu\v{z}ek, Nathaniel Johnston, Ashwin Nayak and Debbie Leung for very delightful discussions.  We thank Andreas Winter for the delight discussion for the Schmidt rank case.

\bibliographystyle{abbrv}
\bibliography{UEX}

\appendix

\section{Proof of Corollary~\ref{cor:lower}}\label{appendix:medium}
Following from Theorem~\ref{thm:concentrate_1}, for any $\lambda>0$, we have
\begin{equation}
\mu\left(P_{min}^{(S_v)}(U)\geq \log d_A-\lambda-\frac{d_A}{d_B \ln 2}\right)\geq 1-  \left(\frac{20\sqrt{2}\log d_A}{\lambda}\right)^{2d_A+2d_B} \exp\left(-\frac{(d_Ad_B-1)\lambda^2}{32\pi^2\ln 2(\log d_A)^2}\right)
\end{equation}
where $U$ is uniformly chosen at random from $\mathcal{U}(d_Ad_B)$ according to the Haar measure and $d_B\geq d_A\geq 3$.

\begin{cor}\label{cor:lower}
$\int\limits_{\mathcal{SU}(d_Ad_B)}P_{min}^{(S_v)}(U)d\mu(U)=\int\limits_{\mathcal{SU}(d_Ad_B)}\min\limits_{\ket{\alpha}, \ket{\beta}} S(\tr_A(U(\ket{\alpha}\otimes\ket{\beta}))) d\mu(U)\geq \log d_A-\beta-1$. Here $\beta=\frac{d_A}{d_B \ln 2}$.
\end{cor}

\begin{proof}
For the integral of $P_{min}^{(S_v)}(U)$ over $\mathcal{SU}(d_Ad_B)$, 

\begin{equation}
\begin{split}
&\int\limits_{\mathcal{SU}(d_Ad_B)}P_{min}^{(S_v)}(U) d\mu(U)\\
={}& \int_{0}^{\infty} \mu\left(P_{min}^{(S_v)}(U)\geq t\right)dt\\
\geq {}&\int_{c}^{\log d_A-\beta} \left(1-  \left(\frac{20\sqrt{2}\log d_A}{\log d_A-\beta-t}\right)^{2d_A+2d_B} \exp\left(-\frac{(d_Ad_B-1)(\log d_A-\beta-t)^2}{32\pi^2\ln 2(\log d_A)^2}\right)\right)dt\\
={}& \log d_A-\beta-c-\int_{c}^{\log d_A-\beta}\left(\frac{20\sqrt{2}\log d_A}{\log d_A-\beta-t}\right)^{2d_A+2d_B} \exp\left(-\frac{(d_Ad_B-1)(\log d_A-\beta-t)^2}{32\pi^2\ln 2(\log d_A)^2}\right)dt\\
={}& \log d_A-\beta-c-\left(20\sqrt{2}\log d_A\right)^{2(d_A+d_B)}\int_{c}^{\log d_A-\beta}t^{-2(d_A+d_B)} \exp\left(-\frac{(d_Ad_B-1)t^2}{32\pi^2\ln 2(\log d_A)^2}\right)dt.\label{eq:central}
\end{split}
\end{equation}

Here $c\in (0,\log d_A-\beta)$ is a feasible solution to $\left(\frac{20\sqrt{2}\log d_A}{c}\right)^{2d_A+2d_B} \exp\left(-\frac{(d_Ad_B-1)c^2}{32\pi^2\ln 2(\log d_A)^2}\right)\leq 1$.

Observe that $g(t)=t^{-2(d_A+d_B)} \exp\left(-\frac{(d_Ad_B-1)t^2}{32\pi^2\ln 2(\log d_A)^2}\right)$ is a monotonically decreasing function over $(0,\log d_A-\beta)$, therefore
\begin{equation}
\begin{split}
&\left|\left(20\sqrt{2}\log d_A\right)^{2(d_A+d_B)}\int_{c}^{\log d_A-\beta}t^{-2(d_A+d_B)} \exp\left(-\frac{(d_Ad_B-1)t^2}{32\pi^2\ln 2(\log d_A)^2}\right)dt\right|\\
\leq{}& \left(20\sqrt{2}\log d_A\right)^{2(d_A+d_B)}c^{-2(d_A+d_B)} \exp\left(-\frac{(d_Ad_B-1)c^2}{32\pi^2\ln 2(\log d_A)^2}\right)(\log d_A-\beta-c).
\end{split}
\end{equation}

One can easily choose $c=1$ when $\left(20\sqrt{2}\log d_A\right)^{2d_A+2d_B} \exp\left(-\frac{(d_Ad_B-1)}{32\pi^2\ln 2(\log d_A)^2}\right)\leq 1$. We can further ask $ \left(20\sqrt{2}\log d_A\right)^{2(d_A+d_B)} \exp\left(-\frac{(d_Ad_B-1)}{32\pi^2\ln 2(\log d_A)^2}\right)(\log d_A-\beta-1)\leq 1$. It is always possible since $\left(20\sqrt{2}\log d_A\right)^{2(d_A+d_B)} \exp\left(-\frac{(d_Ad_B-1)}{32\pi^2\ln 2(\log d_A)^2}\right)(\log d_A-\beta-1)$ tends to $0$ when $d_A$ goes to $\infty$.

As a consequence, when  $d_A$ tends to infinity, $\left(20\sqrt{2}\log d_A\right)^{2(d_A+d_B)} \exp\left(-\frac{(d_Ad_B-1)}{32\pi^2\ln 2(\log d_A)^2}\right)(\log d_A-\beta-1)\leq 1$. We will have
$\int\limits_{\mathcal{SU}(d_Ad_B)}\min\limits_{\ket{\alpha}, \ket{\beta}} S(\tr_A(U(\ket{\alpha}\otimes\ket{\beta}))) d\mu(U)\geq \log d_A-\beta-1$.
\end{proof}

\section{Proof of Lemma~\ref{lemma:closure}}\label{appendix:closure}
\begin{lemma}\label{lemma:closure}
$\dim(\overline{X_r})\leq d_A^2d_B^2-(d_A-r)(d_B-r)+2r-3$, where $\overline{X_r}$ is the Zariski closure of $X_r$.
\end{lemma}

The following technical lemmas will be needed.

\begin{lemma}[\cite{Tauvel:2005us}]\label{lemma:dim1}
$Z_1$ and $Z_2$ are both irreducible varieties over $\mathbb{C}$, and $\phi:Z_1\rightarrow Z_2$ is a dominant morphism, then $\dim(Z_2)\leq \dim(Z_1)$. Here, dominant means $\Phi(Z_1)$ is dense in $Z_2$.
\end{lemma}

\begin{lemma}[\cite{Tauvel:2005us}]\label{lemma:dim2}
$Z_1$ and $Z_2$ are both varieties over $\mathbb{C}$, and $\phi:Z_1\rightarrow Z_2$ is a morphism, then $\dim(Z_1)\leq \dim(Z_2)+\max\limits_{z\in Z_2}{\dim(\phi^{-1}(z))}$.
\end{lemma}

Lemma~$\ref{lemma:dim1}$ and Lemma~\ref{lemma:dim2} establish a connection between the dimensions of domain and codomain of a variety morphism.

\begin{proof}
We have a morphism $F:GL(d_Ad_B)\times \mathbb{P}^{d_Ad_B-1} \rightarrow \mathbb{P}^{d_Ad_B-1}$ which is just the left
action of $GL(d_Ad_B)$ on  $\mathbb{P}^{d_Ad_B-1}$, defined by $F(g,[w])=[g\cdot w]$.

We let $y_0=(1,0,\cdots,0)$ be a row vector with $d_Ad_B$ entries, and for any given $y_1$, $y_2\in \mathbb{P}^{d_Ad_B-1}$, we choose proper $g_1$ and $g_2\in GL(d_Ad_B)$, such that $[g_1\cdot y_0]=[y_1]$ and $[g_2\cdot y_0]=[y_2]$. Then we have 
\begin{equation}
[g\cdot y_2]=[y_1] \iff  [g g_2\cdot y_0]=[g g_1\cdot y_0] \iff [g_1^{-1}gg_2\cdot y_0]=[y_0].
\end{equation}

From above observations, F has the following property: for any $y_1$, $y_2\in \mathbb{P}^{d_Ad_B-1}$,
$F^{-1}(y_2)\cap \{GL(d_Ad_B)\times \{y_1\}\}\cong \left\{\left(\begin{array}{cc}
z_1& \alpha\\
0 & g'
\end{array}\right):z_1\in \mathbb{C}\backslash \{0\}, g'\in GL(d_Ad_B-1), \alpha\in \mathbb{C}^{k-1} \text{ is a row vector}\right\}$. Hence $\dim(F^{-1}(y_2)\cap {GL(d_Ad_B)\times
\{y_1\}})=d_A^2d_B^2-(d_Ad_B-1)$.

Let $P_1$, $P_2$ be projections of $GL(d_Ad_B)\times \mathbb{P}^{d_Ad_B-1}$ to $GL(d_Ad_B)$, $\mathbb{P}^{d_Ad_B-1}$ respectively.
Now we only look at $GL(d_Ad_B)\times \Sigma_{d_A,d_B}\subseteq GL(d_Ad_B)\times \mathbb{P}^{d_Ad_B-1}$, to get
$F:GL(d_Ad_B)\times \Sigma_{d_A,d_B}\rightarrow \mathbb{P}^{d_Ad_B-1}$. Then we have a characterization of $X_r$:
$X_r=P_1F^{-1}\left(\Sigma_{d_A,d_B}^{r-1}\right)$. In fact
\begin{alignat*}{2}
& g\in X_r \\
\iff {}& g(\Sigma_{d_A,d_B})\cap \Sigma_{d_A,d_B}^{r-1}\neq \emptyset  \\
\iff {}&\exists z_1\in \Sigma_{d_A,d_B}, z_2\in \Sigma_{d_A,d_B}^{r-1},&\qquad& \text{s.t.}\; g(z_1)=z_2 \\
\iff {}&\exists z_1\in \Sigma_{d_A,d_B}, z_2\in \Sigma_{d_A,d_B}^{r-1},&\qquad& \text{s.t.}\; (g, z_1)\in F^{-1}(z_2)\\
\iff {}&\exists z_2\in \Sigma_{d_A,d_B}^{r-1},&\qquad& \text{s.t.}\; g\in P_1F^{-1}(z_2)\\
\iff {}& g\in P_1 F^{-1}\left(\Sigma_{d_A,d_B}^{r-1}\right).
\end{alignat*}
So $\overline{X_r}\subseteq GL(d_Ad_B)$ is the Zariski closure of $X_r$, which is also an algebraic variety.

Next, we assert that 
$P_1: F^{-1}\left(\Sigma_{d_A,d_B}^{r-1}\right)\rightarrow \overline{X_r}$ is a dominant morphism.

Furthermore, consider $\Psi: F^{-1}\left(\Sigma_{d_A,d_B}^{r-1}\right)\rightarrow \Sigma_{d_A,d_B}\times \Sigma_{d_A,d_B}^{r-1}$ given by $\Psi(g,[z])=([z],[g\cdot z])$.

For $\forall z_1\in \Sigma_{d_A,d_B}$, $z_2\in \Sigma_{d_A,d_B}^{r-1}$, we have $\Psi^{-1}(z_1,z_2)= (g_2 T g_1^{-1}, z_1)$, where $T=\left\{\left(\begin{array}{cc}
z_0& \alpha\\
0 & g'
\end{array}\right):z_0\in \mathbb{C}\backslash \{0\}, g'\in GL(d_Ad_B-1), \alpha\in \mathbb{C}^{d_Ad_B-1} \text{ is a row vector}\right\}$, and $g_1, g_2\in GL(d_Ad_B)$, s.t. $g_1(y_0)=z_1$, $g_2(y_0)=z_2$.
So this is a dominant morphism. Then we obtain
\begin{align*}
\dim\left(F^{-1}\left(\Sigma_{d_A,d_B}^{r-1}\right)\right)
\leq {}& \dim(T)+\dim\left(\Sigma_{d_A,d_B}\times \Sigma_{d_A,d_B}^{r-1}\right)\\
={}& d_A^2d_B^2-d_Ad_B+1+\dim(\Sigma_{d_A,d_B})+\dim\left(\Sigma_{d_A,d_B}^{r-1}\right).
\end{align*}

It is required in Lemma~\ref{lemma:dim1} that varieties $Z_1$ and $Z_2$ are irreducible. Actually, this condition can be weakened. Lemma~\ref{lemma:dim1} is still true for the more general case that $Z_1$ and $Z_2$ are closed subsets of irreducible varieties~\cite{Hartshorne:1983we}. Through this approach, we can fill out the gap and apply this lemma without danger of confusion. Indeed, the irreducibility of $Z_1$ and $Z_2$ really holds, but the verification is not easy.

Then from Lemma~\ref{lemma:dim1} and Lemma~\ref{lemma:dim2}, we will have
\begin{align*}
\dim(\overline{X_r})\leq {}&\dim\left(F^{-1}\left(\Sigma_{d_A,d_B}^{r-1}\right)\right)\\
\leq {}&
\left(d_A^2d_B^2-(d_Ad_B-1)\right)+\dim(\Sigma_{d_A,d_B})+\dim\left(\Sigma_{d_A,d_B}^{r-1}\right)\\
={}& d_A^2d_B^2-(d_Ad_B-1)+(d_A+d_B-2)+d_Ad_B-(d_A-r+1)(d_B-r+1)-1\\
={}& d_A^2d_B^2-(d_A-r)(d_B-r)+2r-3.
\end{align*}
\end{proof}

\end{document}